\documentclass[11pt]{article}

\usepackage[normalem]{ulem}

\usepackage{amsfonts,amsmath,amsthm,times,fullpage,amssymb,epsfig,graphicx}

\usepackage{enumerate}

\usepackage{hyperref}

\hypersetup{colorlinks=false}

\usepackage[dvipsnames]{xcolor}

\usepackage{color}
\usepackage{dsfont}
\usepackage{wrapfig}
\usepackage{relsize}

\usepackage{url}

\usepackage[left=1.0in,top=1.0in,right=1.0in,bottom=1.0in,nohead,textheight=10in,footskip=0.3in]{geometry}

\usepackage[noend]{algpseudocode} 
\usepackage[nothing]{algorithm}  
\usepackage{caption}
 
\newtheorem{theorem}{Theorem}[section]
\newtheorem{lemma}[theorem]{Lemma}
\newtheorem{corollary}[theorem]{Corollary}

\newtheorem{definition}[theorem]{Definition} 
\newtheorem{remark}{Remark}[section]
\newtheorem{proposition}[theorem]{Proposition}

\newtheorem*{gLLL}{General LLL}

\newtheorem*{cd}{Causality Digraph}
\newtheorem*{causality}{Causality}

\newcommand\ex{{\mathbb{E}}}

\newcommand{\param}{\psi} 

\newcommand{\beq}{\begin{equation}}
\newcommand{\eeq}{\end{equation}}

\newcounter{fooTH}
\newcounter{fooEQ}

\begin{document}

\title
{Commutative Algorithms Approximate the LLL-distribution}

\author{Fotis Iliopoulos
\thanks{Research supported by NSF grant CCF-1514434 and the Onassis Foundation.} \\ 
University of California Berkeley \\
{\small \texttt{fotis.iliopoulos@berkeley.edu}}
}

\maketitle

\begin{abstract}
Following the groundbreaking Moser-Tardos algorithm for the Lov\'{a}sz Local Lemma (LLL), a series of works have exploited a key ingredient of the original analysis, the witness tree lemma, in order to: derive deterministic, parallel and distributed algorithms for the LLL, to estimate the entropy of the output distribution, to partially avoid bad events, to deal with super-polynomially many bad events, and even to devise new algorithmic frameworks. Meanwhile, a parallel line of work has established tools for analyzing stochastic local search algorithms motivated by the LLL that do not fall within the Moser-Tardos framework. Unfortunately, the aforementioned results do not transfer to these more general settings. Mainly, this is because the witness tree lemma, provably, does not longer hold. Here we prove that for commutative algorithms, a class recently introduced by Kolmogorov and which captures the vast majority of LLL applications, the witness tree lemma does hold. Armed with this fact, we extend the main result of Haeupler, Saha, and Srinivasan to commutative algorithms, establishing that the output of such algorithms well-approximates the \emph{LLL-distribution}, i.e., the distribution obtained by conditioning on all bad events being avoided, and give several new applications. For example, we show that the recent algorithm of Molloy for list-coloring number of sparse, triangle-free graphs can output exponential many list-colorings of the input graph.
\end{abstract}

\thispagestyle{empty}

\newpage

\setcounter{page}{1}

\section{Introduction}

Many problems in combinatorics and computer science can be phrased as finding an object that lacks certain bad properties, or ``flaws". In this paper we study algorithms that take as input a flawed object and try to remove all flaws by transforming the object through repeated probabilistic action.
 
Concretely, let $\Omega$ be a  set of objects and let $F = \{f_1, f_2, \ldots, f_m \}$ be a collection of subsets of $\Omega$. We will refer to each $f_i \in F$ as a \emph{flaw} to express that its elements share some negative feature. For example, if a CNF formula $F$ on $n$ variables has clauses $c_1,c_2,\ldots,c_m$, we can define for each clause $c_i$ the flaw (subcube) $f_i \subseteq \{0,1\}^n$ whose elements violate $c_i$. Following linguistic rather than mathematical convention we say that $f$ is present in $\sigma$ if $f \ni \sigma$ and that $\sigma \in \Omega$ is \emph{flawless} (perfect) if no flaw is present in $\sigma$. 

To prove the  \emph{existence} of flawless objects we can often use the Probabilistic Method. As a matter of fact, in many interesting cases, this is the only way we know how to do so. To employ the Probabilistic Method, we introduce a probability measure $\mu$ over $\Omega$ and consider the collection of ``bad" events corresponding to flaws. If we are able to show that the probability to avoid all bad events is strictly positive, then this implies the existence of a flawless object.  A trivial example is the case where all the bad events are independent of one another and none of them has probability one. One of the most powerful tools of the Probabilistic Method is the Lov\'{a}sz Local Lemma~\cite{LLL} which weakens the latter restrictive condition of independence to a condition of limited dependence.

Making the LLL constructive was the study of intensive research for over two decades~\cite{beck_lll,alon_lll,mike_stoc,Czumaj_lll,aravind_08}. The breakthrough was made by Moser~\cite{Moser} who gave a very simple algorithm that finds a satisfying assignment of a $k$-CNF formula, under conditions that nearly match the LLL condition for satisfiability. Very shortly afterwards, Moser and Tardos~\cite{MT} made the general LLL constructive for any \emph{product} probability measure over explicitly presented variables. Specifically, they proved that whenever the general LLL conditon holds, the {\sl Resample} algorithm, which repeatedly selects \emph{any} occurring bad event and resamples all its variables according to the measure, i.e., independently, quickly converges to a flawless object.

The first result that made the LLL constructive in a non-product probability space was due to Harris and Srinivasan in~\cite{PermHarris}, who considered the space of permutations endowed  with the uniform measure. Subsequent works by Achlioptas and Iliopoulos~\cite{AIJACM,Harmonic} introducing the flaws/actions framework, and of Harvey and Vondr\'{a}k~\cite{HV} introducing the resampling oracles framework, made the LLL constructive in more general settings. These frameworks~\cite{AIJACM,Harmonic,HV} provide tools for analyzing \emph{focused} stochastic search algorithms~\cite{papafocus}, i.e., algorithms which, like {\sl Resample}, search by repeatedly selecting a flaw of the current state and moving to a random nearby state that avoids it, in the hope that, more often than not, more flaws are removed than introduced, so that a flawless object is eventually reached. At this point, all LLL applications we are aware of have efficient algorithms analyzable in these frameworks.

Besides conditions for existence and fast convergence to  perfect objects, one could ask further questions regarding properties of focused search algorithms. For instance, ``are they parallelizable ?", ``how many solutions can they output?", ``what is the expected  ``weight" of a solution?", etc. These questions and more have been answered for the Moser-Tardos algorithm in a long series of work~\cite{MT,Haeupler_jacm,EnuHarris,szege_meet,determ,distributed, ParallelHarrisHaeupler,ParallelHarris}. 
As a prominent example,  the   result of Haeupler,  Saha and Srinivasan~\cite{Haeupler_jacm}, as well as follow-up   works of Harris and Srinivasan~\cite{EnuHarris,NewBoundsHarris}, allow one to argue about the dynamics of the MT process, resulting in  several new  applications  such as estimating the entropy of the output distribution, partially avoiding  bad events, dealing with super-polynomially many bad events, and even new frameworks~\cite{PartResmp1,PartResmp2}.

Unfortunately, most of these follow-up results that further enhance, or exploit, our understanding  of the MT process are not transferable to the general settings of~\cite{AIJACM,HV,Harmonic}.  
Mainly, this is because a key and elegant  technical result of the original analysis of Moser and Tardos, \emph{the witness tree lemma}, does not longer hold under the most general assumptions~\cite{HV}.
Roughly, it states that any tree of bad events growing backwards in time from a certain root bad event $A_i$, with the children of each node $A_j$ being bad events that are  adjacent to $A_j$ in the dependency graph, has  probability of being  consistent with the trajectory of the algorithm that is bounded by the product of the probabilities of all events in this tree.
 The witness tree lemma and its variations~\cite{szege_meet,ParallelHarrisHaeupler} has been used for several other purposes besides those already mentioned, such as designing deterministic, parallel and distributed algorithms for the LLL~\cite{MT,determ,distributed,ParallelHarrisHaeupler,ParallelHarris}. 
 
On the other hand,    Harris and Srinivasan~\cite{PermHarris} do manage to prove the witness tree lemma for their algorithm for the LLL on the space of permutations, via an analysis that is tailored specifically to this setting. Although their proof  does not seem to be easily generalizable to general spaces, their success makes it   natural to ask if we can impose mild assumptions in the general settings of~\cite{AIJACM,HV,Harmonic}  under which the  witness tree lemma (and most of its byproducts) can be established.

The main contribution of this paper is to answer this question positively by showing that it is possible to prove the witness tree lemma in  the \emph{commutative setting}.  The latter was recently introduced by Kolmogorov~\cite{Kolmofocs},  who showed that under its assumptions  one can obtain  parallel algorithms, as well as  the flexibility of  having arbitrary flaw choice strategy  in the frameworks of~\cite{AIJACM,HV,Harmonic}. We note that the commutative setting captures the vast majority of LLL applications, including but not limited to  both the variable  and the permutation settings.

Subsequently to the present work, Achlioptas, Iliopoulos and Sinclair~\cite{AIS} gave a simpler proof of the witness tree lemma  under a more general notion of commutativity (essentially matrix commutativity) at the mild cost of slightly restricting the family of flaw choice strategies (as we will see, in this paper the flaw choice strategy can be arbitrary).

Armed with the witness tree lemma, we are able to study  properties of algorithms in the commutative setting  and give several  applications.

\paragraph{Distributional Properties.}

As already mentioned, one of the most important applications of the witness tree lemma is given in the  paper of Haeupler, Saha and Srinivasan~\cite{Haeupler_jacm}, who study  properties of the \emph{MT-distribution}, the output distribution of the MT algorithm. Their main result is that the MT-distribution well-approximates the \emph{LLL-distribution}, i.e., the distribution obtained by conditioning on all bad events being avoided. As an example, an  immediate consequence of this fact  is that one can argue about the expected \emph{weight} of the output of the MT algorithm, given a weighting function over the space $\Omega$.  Furthermore, as shown in the same paper~\cite{Haeupler_jacm} and follow-up papers by Harris and Srinivasan~\cite{EnuHarris,NewBoundsHarris}, one can lower bound the entropy of the MT distribution, go beyond the LLL conditions (if one is willing to only partially avoid bad events), and deal with applications with super-polynomially many  bad events.

Here we extend the result of~\cite{Haeupler_jacm} to the commutative setting: Given a  commutative algorithm that is perfectly compatible with the underlying probability measure, its output well-approximates the LLL-distribution in the same sense  the MT-distribution does in the variable setting. For arbitrary commutative algorithms, the quality of the approximation additionally depends on the compatibility of the algorithm with the measure on the  event(s) of interest. A simplified, and imprecise, version of our main theorem, which assumes that the initial state of the algorithm is sampled according to the underlying probability distribution $\mu$, is as follows.  The formal statement of our main theorem can be found in Section~\ref{Results}. 
 \begin{theorem}[Informal and Imprecise Statement]\label{main_informal}
If algorithm $\mathcal{A}$ is commutative  and the algorithmic LLL conditions hold then,  for each $E \subseteq \Omega$,
\begin{align*}
 \Pr \left[  E    \right]  \le     \gamma(E)     \left( 1 + \frac{ 1}{ d} \right)^{D_E}  \enspace, 
\end{align*}
where $E$ is independent of all but at most $D_E$ flaws,   $d$ is the maximum degree of the dependency graph, $\gamma(E) \ge \mu(E) $ is a measure of the ``compatibility" between $\mathcal{A}$ and the underlying probability distribution $\mu$ at $E$, and $\Pr[E]$ is the probability that $\mathcal{A}$ ever reaches $E$ during its execution.
 \end{theorem}
Moreover, we quantitatively improve the bounds of~\cite{Haeupler_jacm} under the weaker assumptions of \emph{Shearer's condition}~\cite{Shearer}, i.e.,   the most general LLL condition under the assumption that the dependency graph is undirected. This allows us to study distributional properties of commutative algorithms using  criteria that lie between the General LLL and Shearer's condition such as the Clique LLL~\cite{CliqueLLL}.  In Section~\ref{Byproducts} we discuss the byproducts of our main theorem that we will use in our applications. Finally, in  Appendix~\ref{super-poly}  we discuss how one can deal with settings with super-polynomially many flaws.

\paragraph{Algorithmic LLL Without a Slack and Arbitrary Flaw Choice Strategy.}

The works of Achlioptas, Iliopoulos and Kolmogorov~\cite{AIJACM,Harmonic,Kolmofocs} require a multiplicative slack in the generalized  LLL conditions in order to establish  fast convergence to a perfect object. On the other hand, Harvey and Vondr\'{a}k~\cite{HV} dispense with this requirement in the important case of algorithms that are  perfectly compatible with the underlying measure under the mild assumption that the dependency graph is undirected.   

Using the witness lemma, we are able to dispense with the multiplicative slack requirement for arbitrary algorithms in the commutative setting and also have the flexibility of arbitrary flaw choice strategy, as in the result of Kolmogorov~\cite{Kolmofocs}.

\paragraph{Improved Running Time Bounds.}

We are able to improve the running time bounds of Harvey and Vondr\'{a}k~\cite{HV} for commutative algorithms, matching those  of Kolipaka and Szegedy~\cite{szege_meet} for the MT algorithm. Whether this could be done was left as an open question in~\cite{HV}.
We note that while the results of Achlioptas, Iliopoulos and Kolmogorov~\cite{AIJACM,Harmonic,Kolmofocs} also manage to give improved running time bounds  they  require  a multiplicative-slack in the LLL conditions.

\paragraph{Concrete Applications.}

In Section~\ref{Applicatia} we give concrete applications  of commutative algorithms showing new results for the problems of \emph{rainbow matchings, list-coloring} and \emph{acyclic edge coloring}.  Each application is chosen so that it demonstrates specific features of our results. 

The first application is  in the space of matchings of a complete graph. We use this problem as an example that allows us  to show how several byproducts of approximating the LLL-distribution can be applied in a black-box manner to a setting that is not captured either by the variable or the permutation setting, and for which we know~\cite{AIJACM,HV,Kolmofocs}  how to design commutative algorithms that  are perfectly compatible with the uniform measure over the state space.

The second, and perhaps most interesting, application is to show that the algorithm of Molloy~\cite{molloy2017list} for finding  proper colorings in  triangle-free graphs with maximum degree $\Delta$ using $(1 + \epsilon) \frac{\Delta}{\ln \Delta}$ colors,  can actually output exponentially many such colorings with positive probability. First, we show that Molloy's algorithm can be analyzed in the general frameworks of the algorithmic LLL and that it is commutative, a fact that  gives us access to properties of its output distribution. Then, we apply results regarding the entropy of the output of commutative algorithms. We show the following theorem.
\begin{theorem}\label{many_solutions}
For every $\epsilon >0$ there exists $\Delta_{\epsilon} $ such that every triangle-free graph $G$ with maximum degree $\Delta \ge \Delta_{\epsilon} $ has list-chromatic number $\chi_{\ell}(G) \le (1+ \epsilon) \frac{ \Delta }{ \ln \Delta}$. 
Furthermore, if $G$ is a graph on $n$ vertices then, for every $\eta > 0$, there exists an algorithm $\mathcal{A}$  that constructs such a coloring in polynomial time with probability at least $1 - \frac{1}{n^{\eta}} $. In addition, $\mathcal{A}$ is able to output  $\mathrm{e}^{ cn  }  $ distinct list-colorings with positive probability, where $ c >0$ is a constant that depends on $\epsilon$ and $\Delta$. 
\end{theorem}

We emphasize that the algorithm of Molloy is a sophisticated  stochastic local search algorithm whose analysis is far from any standard LLL setting. The fact that our results allow us to state non-trivial facts about its distributional properties almost in a black-box fashion is testament to their flexibility.

In the third application we show how one can use bounds on the output distribution of commutative algorithms that are induced by the Shearer's condition in order to analyze  applications of the Clique version of the Local Lemma in the problem of acyclic edge coloring of a graph.

%
%
%

\section{Background and Preliminaries}\label{Background}

In this section we present the necessary background and definitions to describe our setting. In Subsection~\ref{LLL}  we describe the Lov\'{a}sz Local Lemma. In Subsections~\ref{setting} and~\ref{commutative_setting}  we formally outline the algorithmic assumptions  of~\cite{AIJACM,HV,Harmonic,Kolmofocs}.  In Subsection~\ref{improved} we describe improved Lov\'{a}sz Local Lemma criteria formulated in our setting.

\subsection{The Lov\'{a}sz Local Lemma}\label{LLL}

To prove the  \emph{existence} of flawless objects we can often use the Probabilistic Method. To do so, we introduce a probability measure $\mu$ over $\Omega$ and consider the collection of ``bad" events corresponding to flaws. If we are able to show that the probability to avoid all bad events is strictly positive, then this implies the existence of a flawless object. One of the most powerful tools to establish the latter is the Lov\'{a}sz Local Lemma~\cite{LLL}.
\begin{gLLL}
Let $(\Omega,\mu)$ be a probability space and $\mathcal{A} = \{A_1, A_2,\ldots,A_m\}$ be a set of $m$ (bad) events. For each $i \in [m]$, let $D(i) \subseteq	 [m] \setminus \{i\}$ be such that $\mu(A_i \mid \cap_{j \in S} \overline{A_j}) = \mu(A_i)$ for every $S \subseteq [m] \setminus (D(i) \cup \{i\})$. If there exist positive real numbers $\{\psi_i\}_{i=1}^m$ such that for all $i \in [m]$,
\begin{equation}\label{eq:LLL}
\frac{\mu(A_i)}{\psi_i}  \sum_{ S \subseteq   D(i) \cup \{i \} }  \prod_{j \in S} \psi_j \le 1 \enspace , 
\end{equation}
then the probability that none of the events in $\mathcal{A}$ occurs is at least $\prod_{i=1}^m 1/(1+\psi_i) > 0$. 
\end{gLLL}
\begin{remark}\label{general}
Condition~\eqref{eq:LLL} above is equivalent to the more well-known form $\mu(A_i) \le x_i \prod_{j \in D(i)} (1-x_j)$, 
where $x_i = \psi_i/(1+\psi_i)$. As we will see, formulation~\eqref{eq:LLL} facilitates refinements.
\end{remark}

Let  $G$ be the digraph over the vertex set $[m]$  with an edge from each $i \in [m]$ to each element of $D(i)\cup \{i \}$. We call such a graph a \emph{dependency} graph. Therefore, at a high level, the LLL states that if there exists a sparse dependency graph and each bad event is not too likely, then perfect objects exist.

\subsection{Algorithmic Framework}\label{setting}

Here we describe the class of algorithms we will consider as well as the algorithmic LLL criteria  for  fast convergence to a perfect object.  Since we will be interested in algorithms that search for perfect objects, we sometimes refer to $\Omega$ as  a state space and to its elements as states.

For a state $\sigma$, we denote by $U(\sigma) = \{ j \in [m] \text{ s.t. } f_j \ni \sigma \}$ the set of indices of flaws that are present at $\sigma$. We consider algorithms which at each flawed state $\sigma$ choose an element of $U(\sigma)$ and randomly move to a nearby state in an effort to \emph{address} the corresponding flaw. Concretely, we will assume  that for every flaw $f_i$ and every state $\sigma \in f_i$ there is a probability distribution $\rho_i(\sigma, \cdot) $ with  a non-empty support $A(i,\sigma) \subseteq \Omega$ such that addressing flaw $f_i$ at state  $\sigma$ amounts to selecting the next state $\sigma'$ from $A(i,\sigma)$ with probability $\rho_i(\sigma, \sigma')$. We call $A(i,\sigma)$ the set of \emph{actions} for addressing flaw $f_i$ at $\sigma$ and note that potentially $A(i,\sigma) \cap f_i \ne \emptyset$, i.e., addressing a flaw does not necessarily imply removing it. The actions for flaw $f_i$ form a digraph $D_i$ on $\Omega$ having an arc $\sigma \xrightarrow{i} \sigma'$ for each pair $(\sigma, \sigma') \in f_i \times  A(i,\sigma) $. Let $D$ be the multi-digraph on $\Omega$ that is the union of all $D_i$.


We consider algorithms that start from a state $\sigma \in \Omega$ picked from an initial distribution $\theta$, and then repeatedly pick a flaw that is present in the current state and address it. The algorithm always terminates when it encounters a flawless state. 

To state the algorithmic LLL criteria for fast convergence of such algorithms we need to introduce two key ingredients. The first  one is a notion of \emph{causality} among flaws that will be used to induce a graph over $[m]$,  which will play a role similar to  the one of the dependency graph in the existential Local Lemma formulation.  We note that there is a formal connection between causality graphs and  dependency graphs (for more details see~\cite{HV}).

\begin{causality}
For an arc $\sigma \xrightarrow{i}  \sigma'$ in $D_i$ and a flaw $f_j$ present in $\sigma'$ we say that $f_i$ causes $f_j$ if $f_i = f_j$ or $f_j \not\ni \sigma$. If $D_i$ contains \emph{any} arc in which $f_i$ causes $f_j$ we say that $f_i$ \emph{potentially causes} $f_j$.
\end{causality}
\begin{cd}
Any digraph $C=C(\Omega,F,D)$ on $[m]$ where $i \rightarrow j$ exists whenever  $f_i$ potentially causes $f_j$ is called a  \emph{causality digraph}. The \emph{neighborhood} of a flaw $f_i$  in $C$  is $\Gamma(i) =\{j : i \to  j \text{  exists in $C$}\}$.
\end{cd}
The second ingredient  is a measure of \emph{compatibility} between the  actions of the algorithm for addressing each flaw $f_i$  (that is, digraph $D_i$) and the probability measure $\mu$ over $\Omega$ which we will use for the analysis. As was shown in~\cite{HV,Harmonic,Kolmofocs} one can capture compatibility by letting
\begin{align}\label{charge}
d_i   =  \max_{\sigma \in \Omega } \frac{\nu_i(\sigma)}{ \mu(\sigma) } \ge 1  \enspace ,
\end{align} 
where $\nu_i(\sigma)$ is the probability of ending up at state $\sigma$ at the end of the following experiment: sample $\omega \in f_i$ according to $\mu$ and address flaw $f_i$ at $\omega$. An algorithm achieving perfect compatibility for flaw $f_i$, i.e., $d_i = 1$, is a \emph{resampling oracle} for flaw $f_i$  (observe that the   Moser-Tardos algorithm is trivially a resampling oracle for every flaw). More generally, ascribing to each flaw $f_i$ the \emph{charge}
\begin{eqnarray*}
\gamma(f_i) =  d_i \cdot \mu(f_i)  =  \max_{\sigma' \in \Omega } \frac{ 1 }{ \mu(\sigma')  }  \sum_{ \sigma \in f_i } \mu(\sigma) \rho_i(\sigma,\sigma') \enspace,
\end{eqnarray*}
yields the following algorithmization condition. If for every flaw $f_i \in F$,
\begin{align} \label{AlgoLLL}
\frac{\gamma(f_i) }{ \psi_i} \sum_{ S \subseteq \Gamma(i) } \prod_{j \in S } \psi_j  < 1   
\end{align}
then  there exists a flaw choice strategy under which the algorithm will reach a perfect object fast. (In most applications, that is in $O \left( \log |\Omega| +  m \max_{i \in [m ] } \log_2 \left( 1 + \psi_{i} \right) \right) $ steps with high probability.)

Throughout the paper we assume that we are given an undirected causality  graph $C$ (and thus the relation $\Gamma( \cdot) $ is symmetric) and we will sometimes write $i \sim j$ if $ j \in \Gamma(i) \leftrightarrow j \in \Gamma(j)$.  Furthermore,  for a set $S \subseteq [m] $ we define $\Gamma(S) = \bigcup_{i \in S} \Gamma(i)  $. Finally, we denote by $\mathrm{Ind}(S) = \mathrm{Ind}_C(S)$ the set of independent subsets of $S$ with respect to $C$.

\subsection{Commutativity}\label{commutative_setting}

We will say that $\sigma \xrightarrow{i} \sigma'$ is a \emph{valid trajectory} if it is possible to get from state $\sigma$ to state $\sigma'$ by addressing flaw $f_i$ as described in the algorithm, i.e., if two conditions hold: $i \in U(\sigma) $ and $\sigma' \in A(i,\sigma) $. Kolmogorov~\cite{Kolmofocs}  described the following \emph{commutativity} condition. We call the setting in which Definition~\ref{commutativity} holds  \emph{the commutative setting}.
\begin{definition}[Commutativity \cite{Kolmofocs}]\label{commutativity}
 A tuple $(F, \sim, \rho) $ is called \emph{commutative} if  there exists a mapping $\mathrm{Swap}$ that sends  any trajectory $\Sigma = \sigma_1 \xrightarrow{i} \sigma_2 \xrightarrow{j} \sigma_3$ with $ i \nsim j$ to another valid trajectory $ \mathrm{Swap}(\Sigma) = \sigma_1 \xrightarrow{j} \sigma_2' \xrightarrow{i} \sigma_3 $, and:
\begin{enumerate}
\item $\mathrm{Swap} $ is injective, 
\item $\rho_i(\sigma_1,\sigma_2) \rho_j( \sigma_2, \sigma_3) =   \rho_j(\sigma_1,\sigma_2') \rho_i( \sigma_2', \sigma_3)   $ \enspace.
\end{enumerate}
\end{definition}
It is straightforward to check that the Moser Tardos algorithm satisfies the commutativity condition. Furthermore, Kolmogorov showed that the same is true for resampling oracles in the permutation~\cite{PermHarris} and perfect matchings~\cite{HV} settings, and Harris~\cite{ParallelHarris} designed commutative resampling oracles for hamiltonian cycles.

Finally, as already mentioned, Kolmogorov showed that  in the  commutativity setting one may choose an arbitrary flaw choice strategy which is a function of the entire past execution history. The same will be true for our results, so we make the convention that given a tuple $(F, \sim, \rho) $ we always fix some arbitrary flaw choice strategy to get a well-defined, commutative algorithm $\mathcal{A} = (F, \sim,\rho)$.

\subsection{Improved LLL Criteria}\label{improved}

Besides the general form of the LLL~\eqref{eq:LLL} there exist improved criteria that apply in the full generality of the LLL setting. The most well-known are the \emph{cluster expansion condition}~\cite{Bissacot} and the \emph{Shearer's condition}~\cite{Shearer}.
Both of these criteria apply when the dependency graph is undirected and have been made constructive~\cite{szege_meet,PegdenIndepen,AIJACM,HV,Harmonic,Kolmofocs} in the most general algorithmic LLL settings. 
\paragraph*{Cluster Expansion Condition.}
The cluster expansion condition strictly improves upon the General LLL condition~\eqref{eq:LLL} by taking advantage of the local density of the dependency graph. 
\begin{definition}
Given a sequence of positive real numbers $\{ \psi_i \}_{i=1}^{m}$, we say that the cluster expansion condition is satisfied if for each $i \in [m]$:
\begin{align}\label{ClusterLLL}
\frac{\gamma(f_i) }{\psi_i }  \sum_{ S \in \mathrm{Ind}( \Gamma(i) ) } \prod_{j \in S} \psi_j \le 1 \enspace.
\end{align}
\end{definition}
\paragraph*{Shearer's Condition.} 
%
%
Let $\gamma \in \mathbb{R}^{m}  $ be the real vector such that $\gamma_i = \gamma(f_i)$.  Furthermore, for $S \subseteq [m]$ define  $\gamma_S = \prod_{j \in S } \gamma_j$ and the polynomial $q_S$:
\begin{align*}
q_S  = q_S(\gamma) = \sum_{ \substack{ I  \in \mathrm{Ind}([m])   \\ S \subseteq I} } (-1)^{ |I|- |S| }\gamma_{I}   \enspace.
\end{align*}
\begin{definition}\label{ShearerLLL}
We say that the Shearer's condition is satisfied if $q_S(\gamma) \ge 0 $ for all $S \subseteq [m]$, and $q_{\emptyset}(\gamma) > 0 $. 
\end{definition}

\section{Statement of Results}\label{Results}

Assuming that the LLL conditions~\eqref{eq:LLL} hold, the \emph{LLL-distribution}, which we denote by $\mu_{  \mathrm{LLL} } $,  is defined as the distribution induced by the measure $\mu$ conditional on no bad event  occurring. The following proposition relates the LLL-distribution to measure $\mu$ making it a powerful tool  that can be used to argue about properties of flawless objects. The idea is that if an (not necessarily bad) event $E$ is independent from most bad events, then its probability under the LLL-distribution is not much larger than its probability under the probability measure $\mu$.

\begin{proposition}[\cite{Haeupler_jacm}]\label{LLL-distribution}
If the LLL conditions~\eqref{eq:LLL}   hold, then for any event $E$: 
\begin{align}\label{LLL_dist_aprox}
\mu_{  \mathrm{LLL} }(E)  \le \mu(E) \sum_{S \subseteq D (E) } \prod_{j \in S}  \psi_j \enspace,
\end{align}
where $D(E) \subseteq [m]$ is such that $\mu( E \mid \bigcap_{j \in S} \overline{A}_j ) = \mu(E)$ for all $S \subseteq [m] \setminus D(E)$.
\end{proposition}

The main result of Haeupler, Saha and Srinivasan~\cite{Haeupler_jacm} is that the Moser-Tardos algorithm approximates well the LLL-distribution, in the sense that the left-hanside of~\eqref{LLL_dist_aprox}  bounds the probability that it ever reaches a subspace $E \subseteq \Omega$   during its execution. Building on this fact,~\cite{Haeupler_jacm} and followup works~\cite{EnuHarris,NewBoundsHarris} manage to show several new applications.

Here we extend the latter result to arbitrary commutative algorithms. Given an arbitrary set $E \subseteq \Omega$ and a commutative algorithm $\mathcal{A}$, consider an extension, $\mathcal{A}_E$, of  $\mathcal{A}$  by defining an extra flaw $f_{m+1}\equiv E$ with its own set of probability distributions $\rho_{m+1}(\sigma, \cdot), \sigma \in E$.  If $\mathcal{A}$ is commutative with respect to $\sim$, we will say that $\mathcal{A}_E $ is a \emph{commutative extension} of $\mathcal{A}$ if $\mathcal{A}_E = (F \cup \{ m+1 \}, \sim, \rho)$ is also commutative.
 
 Commutative extensions should be interpreted  as a tool to bound the probability  that $\mathcal{A}$ ever reaches a subset $E$ of the state space.  That is, they are defined only for the purposes of the analysis and,  typically in  applications, they are a natural extension of the algorithm. For example, in the case of the Moser-Tardos algorithm applied to $k$-SAT, if one would like to bound the probability that the algorithm ever reaches a state such that variables $x_1, x_2$ of the formula are both set to true,  then one could define $f_{m+1} = \{ \sigma \in \Omega \text{ s.t. } \sigma(x_1) = \sigma(x_2) = 1\}$ along with the corresponding commutative extension of the Moser-Tardos algorithm that addresses $f_{m+1}$ by resampling variables $x_1, x_2$ according to the product measure over the variables of the formula that the Moser-Tardos algorithm uses whenever it needs to resample a violated clause. Indeed, commutative extensions of this form are implicitly defined in the analysis of~\cite{Haeupler_jacm} for the Moser-Tardos algorithm. 
  
We will use the notation $\Pr[\cdot] = \Pr_{\mathcal{A}}[ \cdot ]$ to refer to the probability of events in the probability space induced by the execution of algorithm $\mathcal{A}$. For example, the probability that $\mathcal{A}$ ever reaches a set $E \subseteq \Omega$ of the state space during its execution will be denoted by $\Pr[E]$. 
\begin{theorem}\label{main}
If $\mathcal{A} = (F,\sim,\rho)$ is commutative and the cluster expansion condition is satisfied then:
\begin{enumerate}
\item  for each  $i \in [m] $: $\ex[N_i]  \le  \lambda_{\mathrm{init}}  \psi_i       $ ;
\item  for each $E \subseteq \Omega$: $\Pr \left[  E    \right]  \le     \lambda_{\mathrm{init}}   \gamma(E)     \sum \limits_{S \in \mathrm{Ind} \left(\Gamma(E) \right)  }  \prod_{ j \in S}  \psi_j       $ ;
\end{enumerate}
where $N_i$ is the number of times flaw $f_i$ is addressed during the execution  of $\mathcal{A}$, $\lambda_{\mathrm{init}} =  \max_{\sigma \in \Omega}\frac{ \theta(\sigma) }{ \mu(\sigma) } $, and $\Gamma(E)$ and $\gamma(E)$ are defined with respect to a fixed commutative extension $\mathcal{A}_E$.
 \end{theorem}
\begin{corollary}
Algorithm $\mathcal{A}$ terminates after  $O( \lambda_{\mathrm{init}} \sum \limits_{i \in [m] } \psi_i  )$  steps in expectation. 
\end{corollary}
\begin{remark}\label{ShearerRemark}
If the Shearer's condition is satisfied, then one can replace $\psi_i$  in Theorem~\ref{main} with $\frac{ q_{\{ i \}}(\gamma) }{  q_{\emptyset}( \gamma) }$.
\end{remark}

We note that the  first part of Theorem~\ref{main} allows us  to guarantee fast convergence of $\mathcal{A}$ to a perfect object without having to assume a ``slack" in the cluster expansion and Shearer's conditions (unlike the works of~\cite{Harmonic,Kolmofocs})  and, moreover, improves upon the (roughly quadratically  worse) running bound of~\cite{HV}, matching the one of~\cite{szege_meet}. Whether the latter could be done was left as an open question in~\cite{HV}.

\section{Proof of Main Results}\label{Proofs}
In this section we state and prove the witness tree lemma for our setting. We then use it to prove Theorem~\ref{main}.
\subsection{The Witness Tree Lemma}\label{WTL}
Given a trajectory $\Sigma = \sigma_1 \xrightarrow{w_1} \ldots \sigma_t  \xrightarrow{w_t} \sigma_{t+1} $ we denote by $W(\Sigma) = ( w_1, \ldots, w_t   ) $ the  \emph{witness sequence} of $\Sigma$. (Recall that according to our notation, $w_i$ denotes  the index of the flaw that was addressed at the $i$-th step). 

To state the witness tree lemma, we will  first need to recall the definition of witness trees from~\cite{MT}, slightly reformulating to fit our setting.
A witness tree $\tau  = ( T, \ell_{T} )$ is a finite rooted, unordered,  tree $T$ along with a labelling $\ell_T: V(T) \rightarrow [m] $ of its vertices with indices of flaws such that the children of a vertex $v \in V(T)  $ receives labels from $\Gamma( \ell(v) ) $. To lighten the notation, we will sometimes write $(v)$ to denote $\ell(v)$ and $V(\tau)$ instead of $V(T)$.  Given a witness sequence $W = (w_1, w_2, \ldots, w_t) $ we associate with each $i \in [t] $ a witness tree $\tau_W(i) $ that is constructed as follows:
Let $\tau_W^{i}(i) $ be an isolated vertex labelled by $w_i$. Then, going backwards for each $j = i-1, i-2, \ldots, 1$: if there is a vertex $v \in \tau_W^{j+1}(i) $  such that $(v)  \sim  w_j $ then we choose among those vertices the one having the maximum 
distance (breaking ties arbitrarily) from the root and attach a new child vertex $u$ to $v$ that we label $w_j$ to get $\tau_W^{j}(i)$. If there is no such vertex $v$ then $\tau_W^{j+1}(i) = \tau_W^{j}(i)$. Finally, let $\tau_W(i) = \tau_W^{1}(i)$.

We will say that a witness tree $\tau$ occurs in  a trajectory $\Sigma$ if  $W(\Sigma) = (w_1, w_2, \ldots, w_t)$ and there is $k \in [t]$ such that $\tau_W(k) = \tau$.  

\begin{theorem}[The witness tree lemma]\label{WitnessTreeLemma}
Assume that $\mathcal{A} = (F,\sim,\rho )$ is commutative. Then, for every witness tree $\tau$ we have that:
\begin{align*}
\Pr[\tau] \le \lambda_{\mathrm{init} } \prod_{v \in V(\tau) }  \gamma(f_{(v) })  \enspace.
\end{align*}
\end{theorem}
We show the proof of Theorem~\ref{WitnessTreeLemma} in Section~\ref{WTL_Proof_Sec}.

\subsection{Witness Trees and Stable Witness Sequences}

Here we prove some properties of witness trees (which are induced by witness sequences of the algorithm) that will be useful to us later. We  also draw a connection between witness trees and \emph{stable witness  sequences}, which we will need  in the proof of Theorem~\ref{main}.  Stable witness sequences were first introduced in~\cite{szege_meet} to make the Shearer's condition constructive in the variable setting.

\subsubsection{Properties of Witness Trees}

The following propositions capture the  main properties of witness trees we will need.

\begin{proposition}\label{indep_levels}
For a witness tree  $\tau =(T,\ell_{T}) $ let $L_i = L_i(\tau)$ denote the set of labels of the nodes at distance $i$ from the root.  For each $i \ge 0$, $L_i \in \mathrm{Ind}([m])$. 
\end{proposition}

\begin{proof}
We will show that for each $i \ge 0$,  and each $\alpha,\beta \in L_i$ we have that $\alpha \nsim \beta$.

Let $W = (w_1, w_2, \ldots, w_t)$ be a witness sequence that can occur in an execution of our algorithm. Let $\alpha, \beta $ be two distinct elements of $L_i$. By the definition of $\tau$, labels $\alpha, \beta$ correspond to two indices $w_{j_1}, w_{j_2} $ of $W$. Assume without loss of generality that $j_1 < j_2$.  Then, according to the algorithm for constructing $\tau$, index  $w_{j_2}$ is ``attached first" to the $i$-th level of $\tau$. The proof is concluded by noticing that if $w_{j_1} = \alpha \sim \beta =  w_{j_2}$ then the node corresponding to $w_{j_1}$ is eligible to be a child of the node corresponding to $w_{j_2}$ during the construction of $\tau$ and, thus, $\beta \notin L_i$, which is a contradiction.

\end{proof}

\begin{proposition}\label{uniqueness}
For a witness sequence $W $ of length $t$  and any two distinct  $i,j \in [t]$ we have that $\tau_W(i) \ne \tau_W(j)$.
\begin{proof}
Let $W = (w_1, w_2, \ldots, w_t )$. Assume w.l.o.g. that $i < j$. If $w_i \ne w_j$ then the claim is straightforward because the root of $\tau_W(i)$  is $w_i$  while the root of $\tau_W(j)$ is $w_j$.   If  $w_i = w_j =  w$, then there two cases. In the first case,  $ w  \in \Gamma(w)$,  and so tree  $\tau_W(j)   $ has at least one more vertex than $\tau_W(i)$. In the second case $ w \notin \Gamma(w) $. This implies that at the $i$-th step of any trajectory $\Sigma$ such that $W(\Sigma)  = W$, flaw $f_{w}$ was addressed and removed.  However, the fact that $w_j = w$ implies that there has to be $ k \in \{i+1, \ldots, j-1\}$ such that addressing $w_k$ introduced $w$ and thus,  $w_k \sim w$. Again, this means that $\tau_W(j)   $ has at least one more vertex than $\tau_W(i)$.

\end{proof}
 
\end{proposition}

\subsubsection{Stable Witness Sequences}\label{stable_set_sequences}

We will now recall the definition of \emph{stable  sequences}~\cite{szege_meet}, which  have   been used in~\cite{HV,Kolmofocs} to make the Shearer's condition constructive.

\begin{definition}
A sequence of subsets $ (I_1, \ldots, I_k )$  of $[m]$  with $ k \ge 1 $ is called  \emph{stable} if 

\begin{enumerate}
\item $I_r \in \mathrm{Ind}([m]) \setminus \{  \emptyset \} $ for each $ r \in [k] $;
\item $I_{r+1} \subseteq \Gamma(I_r) $ for each $ r \in [k-1]$.
\end{enumerate}

\end{definition}

\begin{definition}\label{stable_seq}
A witness sequence $W =  (w_1, \ldots, w_t )$ is called \emph{stable} if it can be partitioned into non-empty sequences as $W =  (W_1, \ldots, W_k )$ such that the elements of each sequence $W_r$ are distinct, and the sequence $\phi_W := (I_1, \ldots, I_k )$ is  stable, where $I_r$ is the set of indices of flaws in $W_r$ (for $r \in [k])$.

 For any arbitrary ordering $\pi$ among indices of flaws, if in addition each sequence $W_r = ( w_i, \ldots, w_j) $ satisfies $w_i \prec_{\pi} \ldots \prec_{\pi} w_j $ then $W$ is called \emph{$\pi$- stable}.
\end{definition}

\begin{proposition}(\cite{Kolmofocs})\label{partitioning}
For a stable witness  sequence the partitioning in Definition~\ref{stable_seq} is unique. 
\end{proposition}
\begin{proof}[Proof Sketch]
Let $W = (w_1, \ldots, w_t)$ be a stable sequence and consider the following algorithm. We start with a single segment containing $w_1$. For $i = 2 $ to $t$, if there exists index $w_k$ in the currently last segment such that $w_k \sim w_i$ then we start a new segment containing $w_i$. Otherwise, we add $w_i$ to the currently last segment. 
\end{proof}
For a witness  sequence  $W=  (w_1, \ldots,  w_t )$ let $\mathrm{Rev}[W]=  (w_t, \ldots, w_1) $ denote the reverse sequence.  Let also $R_W$ denote the first set (the ``root") of the stable sequence $\phi_W := (I_1, \ldots, I_k) $, i.e., $R_W = I_1$. Finally, let $\mathcal{R}_i^{\pi}$ be the set of    witness sequences $W$ such that  $\mathrm{Rev}[W]$ is $\pi$-stable and $R_{ \mathrm{Rev}[ W]  } = \{ i \}$.

There is a connection between stable  sequences and witness trees that we will need for the proof of Theorem~\ref{main}  and which we will describe below. 

Let $\mathcal{W}_i$ denote the set of witness trees with root labelled by $i$. For each $\tau \in \mathcal{W}_i$, let  $\chi_{\pi}(\tau)$  be the  \emph{ordered} witness tree that is induced by ordering the children of each node in $\tau$ from left to right, increasingly according to $\pi$. Define $\mathcal{W}_i^{\pi} := \chi_{\pi} ( \mathcal{W}_i  ) $ and observe that  $\chi_{\pi}$ is a bijection. Finally, recall that for a witness tree $\tau$ we denote by $L_j(\tau)$ the set of labels of the  nodes at distance $j$ from the root.

\begin{lemma}\label{trees_stable}
There is a bijection $\chi_i^{\pi} $ mapping $\mathcal{R}_i^{\pi}$ to $\mathcal{W}_i^{\pi}$ with the following property:  Fix $W \in \mathcal{R}_i^{\pi}$ and let  $(I_1 = \{i \}, I_2, \ldots, I_k )$  be the unique partitioning of $\mathrm{Rev}[W]$ guaranteed by Proposition~\ref{partitioning}. Then,  $ I_{j} = L_{j-1} \left(\chi_i^{\pi} \left( W \right) \right) $ for each $j \in [k]$ .
\end{lemma}
\begin{proof}
 Consider a witness sequence $W \in \mathcal{R}_i^{\pi}$ of length $t$. We define $\chi_i^{\pi}(W)  :=   \chi_{\pi} (\tau_{W}(t) ) $. That is, we map $W$ to the $\pi$-ordered  witness tree that is induced by applying the procedure that constructs witness trees to the final element of the sequence. Recall now the procedure in the proof of Proposition~\ref{partitioning}. The key observation is that the application of this procedure  to $\mathrm{Rev}[W]$ is ``equivalent" to the procedure that constructs $\tau_{W}(t)$, in the sense that the  decisions taken for partitioning $\mathrm{Rev}[W]$ to segments by the procedure of Proposition~\ref{partitioning} are identical to the decisions taken by the procedure that constructs $\tau_W(t)$  in order to form $L_j(\tau_{W}(t))$, $j \ge 0$. In particular, if $\phi_{ \mathrm{Rev}[W] } = (I_1 = \{i \}, I_2, \ldots, I_k )$ is the unique partitioning of $\mathrm{Rev}[W]$, then $ I_j =   L_{j-1} ( \tau_W(t) ) = L_{j-1} (\chi_i^{\pi}(W) ) $ for each $j \in [k]$.

It remains to show that $\chi_i^{\pi}$ is bijection. To see this, at first observe that from $\chi_i^{\pi}(W)$ one can uniquely reconstruct $\phi_{\mathrm{Rev}[W]  } $. Given $\phi_{\mathrm{Rev}[W]  } $ one can reconstruct $\mathrm{Rev}[W]$ (and, thus, $W$) by ordering each segment of $\phi_{\mathrm{Rev}[W]  } $ according to $\pi$.

\end{proof}

\subsubsection{Counting Witness Trees}

In our proofs we will need to bound the sum over all trees $\tau \in \mathcal{W}_i$ of the product of charges of the labels of the nodes of each tree $\tau$.  Fortunately, the method for doing that is well trodden by now (see for example ~\cite{MT,PegdenIndepen}). Here we show the following   lemma whose proof can be found in  Appendix~\ref{Misc}. Recall that $\mathcal{W}_i$ denotes the set of all possible witness trees with root that is labelled by $i$.   
\begin{lemma}\label{witness_trees_sum}
 If the cluster expansion condition is satisfied  then:
\begin{align*}
\sum_{ \tau \in \mathcal{W}_i  }  \prod_{ v \in V(\tau) } \gamma \left(  f_{(v  )}  \right) \le \psi_i \enspace.
\end{align*}
\end{lemma}

We also show the  following lemma that can  be used whenever the Shearer's condition applies.
\begin{lemma}\label{ShearerTreeCounting}
If the Shearer's condition is satisfied then:
\begin{align*}
\sum_{ \tau \in \mathcal{W}_i  }  \prod_{ v \in V(\tau) } \gamma \left(  f_{ (v)  }  \right) \le \frac{ q_{  \{ i \} } (\gamma) }{q_{\emptyset} (\gamma) }   \enspace.
\end{align*}
\end{lemma}

\begin{proof}
We first observe that due to Lemma~\ref{trees_stable} we have that  
\begin{align*}
\sum_{ \tau \in \mathcal{W}_i  }  \prod_{ v \in V(\tau) } \gamma \left(  f_{ (v  )}  \right) =     \sum_{ \tau \in \mathcal{W}_i^{\pi}  }  \prod_{ v \in V(\tau) } \gamma \left(  f_{ (v  )}  \right) =  \sum_{ W \in \mathcal{R}_i^{\pi}  }  \prod_{ w \in W } \gamma \left(  f_{w}  \right)  \enspace.
\end{align*}

Now let $\mathrm{Stab}_i$ denote the set of stable set sequences whose first segment is $\{i \}$ and also every segment is non-empty. For $\phi = (I_1, I_2, \ldots, I_k )  \in \mathrm{Stab}_i$ define $\gamma_{\phi}  = \prod_{i=1}^{k} \prod_{i \in I} \gamma(f_{i} )  $.
Observe  that there is a natural injection from  $\mathcal{R}_i^{\pi}$ to $\mathrm{Stab}_i$  which maps each sequence $W \in \mathcal{R}_i^{\pi}$ to $\phi_{\mathrm{Rev}[W] } $.
This is because given $\phi_{\mathrm{Rev}[W] } $ one can reconstruct $\mathrm{Rev}[W]$ (and, thus, $W$) by ordering each segment of $\phi_{\mathrm{Rev}[W]  } $ according to $\pi$. The latter observation implies that:
\begin{align*}
\sum_{ \tau \in \mathcal{W}_i  }  \prod_{ v \in V(\tau) } \gamma \left(  f_{ (v  )}  \right) \le \sum_{ \phi \in \mathrm{Stab}_i } \gamma_{\phi}   =  \frac{ q_{ \{ i \} } (\gamma) }{ q_{\emptyset}(\gamma)}   \enspace, 
\end{align*}
where the proof of the last inequality  can be found in Theorem 14 of~\cite{szege_meet} and in Lemmas 5.26, 5.27 of~\cite{HV}. 
\end{proof}
\begin{remark}
We note that if we have assumed a stronger ``cluster expansion condition" (namely, in~\eqref{ClusterLLL} we have $\Gamma(i) \cup \{ i \}$) instead of $\Gamma(i)$, then Corollary~\ref{witness_trees_sum} could have also been shown as an immediate application of Lemma~\ref{ShearerTreeCounting}, since it is  known (\cite{Bissacot,HV,Kolmofocs}) that,  in this case,  for every $i \in [m]$ we have that $\frac{q_{ \{  i \} }(\gamma)}{q_{\emptyset}(\gamma) } \le \psi_i  $. 
\end{remark}
\subsection{Proof of Theorem~\ref{main}} 

We first prove Theorem~\ref{main}. The first part follows by observing that if $W$ is the witness sequence corresponding to the trajectory of the algorithm, then $N_i$ is the number of occurrences of flaw $f_i$ in $W$, and according to Proposition~\ref{uniqueness}, also the number of distinct witness trees occurring that have their root labeled $i$. Therefore,  one can bound the expectation of $T_i$ by summing the bounds in Lemma~\ref{WitnessTreeLemma}. In particular, Lemma~\ref{witness_trees_sum} concludes the proof of the first part.

To see  the second part of Theorem~\ref{main}, consider the new set of flaws  $F' = F \cup \{f_{m+1} \}$, where $f_{m+1} = E$, as well as a ``truncated" commutative extension  $\mathcal{A}'$ of $\mathcal{A}$ with  the following properties:

\begin{enumerate}[(i)]
\item For each state $\sigma \notin f_{m+1} $ algorithm $\mathcal{A}'$  invokes $\mathcal{A}$ to choose its next state;
\item $\gamma(E) := \gamma_{\mathcal{A}'}(f_{m+1} )$;
\item $f_{m+1}$ is always of the highest priority:  when at a state $\sigma \in f_{m+1} $, $\mathcal{A'}$ chooses to address $f_{m+1}$;
\item $\mathcal{A}'$ stops after it addresses $f_{m+1}$ for the first time.
\end{enumerate}

\noindent
By coupling $\mathcal{A}$ and $\mathcal{A}'$ we see that  $\Pr_{ \mathcal{A} }[ E ] =  \Pr_{\mathcal{A}'}[f_{m+1} ] $. Let $\mathcal{W}_{E}$ be the set of witness trees that can occur in an execution of $\mathcal{A}'$ and whose root is labelled  by $m+1$. 
Notice that, due to property (iv) of $\mathcal{A}'$, every tree  $\tau  \in \mathcal{W}_E$ contains exactly one node (the root) labelled by $m+1$, while every other node is labelled by elements in $[m]$. Furthermore,  the set of labels of the children of the root of $\tau$ is an element of $\mathrm{Ind}( \Gamma(E) ) $. Finally, if  $v$ is a node that corresponds to a child of the root in $\tau$, then the subtree $\tau_v$  that is rooted at $v$ is an element of $\mathcal{W}_{(v)}$.
Using Theorem~\ref{WitnessTreeLemma} and the fact that $\mathcal{A}'$ is commutative we get:
\begin{align*}
\Pr_{ \mathcal{A} }[ E ]  \le  \sum_{ \tau \in \mathcal{W}_E } \Pr_{\mathcal{A}'}[\tau]  \le   \lambda_{\mathrm{init}}  \gamma(E)   \sum_{ S \in \mathrm{Ind}(\Gamma(E) )  }  \left(  \prod_{j \in S}  \sum_{ \tau \in \mathcal{W}_j} \prod_{v \in \tau} \gamma((v))  \right)    \le  \lambda_{\mathrm{init}}  \gamma(E) \sum_{ S \in \mathrm{Ind}(\Gamma(E) )  } \prod_{j \in S}  \psi_j   \enspace,
\end{align*}
where the last equality follows from Lemma~\ref{witness_trees_sum}. The proof of Theorem~\ref{main} in the Shearer's condition regime is  identical,  where instead of Lemma~\ref{witness_trees_sum} we use Lemma~\ref{ShearerTreeCounting}. 
 \subsection{Proof of Lemma~\ref{WitnessTreeLemma}}\label{WTL_Proof_Sec}

Throughout the proof, we will use ideas and definitions  from~\cite{Kolmofocs}.    We also  note that we will  assume w.l.o.g. that  algorithm $\mathcal{A}$ follows a deterministic  flaw choice strategy. This is because randomized flaw choice strategies can equivalently be interpreted as convex combination of deterministic ones (and therefore,  randomized strategies can be seen as taking expectation over deterministic ones).

For a trajectory $\Sigma$ of length $t$ we define 
\begin{align*}
p(\Sigma) = \lambda_{\mathrm{init}} \prod_{ i =1  }^t \rho_{w_i}(\sigma_i, \sigma_{i+1} ) 
\end{align*}
and notice that $\Pr[\Sigma] \le p(\Sigma)$. Furthermore, we say that a trajectory $\Sigma'$ is a proper prefix of $\Sigma$ if $\Sigma'$ is a prefix of $\Sigma$ and $\Sigma \ne \Sigma'$.  
\begin{definition}[\cite{Kolmofocs}]\label{valid_traj}
A set $\mathcal{X}$ of trajectories of the algorithm will be called  \emph{valid} if (i)  all trajectories in $\mathcal{X}$ follow the same deterministic flaw choice strategy (not necessarily the same used by $\mathcal{A}$); and (ii) for any $\Sigma ,\Sigma' \in \mathcal{X}$  trajectory $\Sigma$ is not a proper prefix of $\Sigma'$.
\end{definition}
 \begin{lemma}[\cite{Kolmofocs}] \label{step_prefixes}
Consider a witness sequence $W =  (w_1,  \ldots, w_t)  $ and a valid set of trajectories $\mathcal{X}$ such that $W$ is a prefix of $W(\Sigma)$ for every $ \Sigma \in \mathcal{X}$. We have that
\begin{align*}
\sum_{ \Sigma  \in \mathcal{X}} p(  \Sigma )  \le  \lambda_{\mathrm{init}} \prod_{i=1}^{t} \gamma( f_{w_i}  )    \enspace.
\end{align*}
\end{lemma}
A  \emph{swap}  is the operation of transforming a trajectory  $\Sigma = \ldots \sigma_1 \xrightarrow{i} \sigma_2 \xrightarrow{j} \sigma_3 \ldots $, with $i \nsim j$, to a trajectory 
$\Sigma' = \ldots \sigma_1 \xrightarrow{j} \sigma_2 ' \xrightarrow{i} \sigma_3 \ldots $, where $   \sigma_1 \xrightarrow{j} \sigma_2 ' \xrightarrow{i} \sigma_3 = \mathrm{Swap}( \sigma_1 \xrightarrow{i} \sigma_2 \xrightarrow{j} \sigma_3 )$. A mapping $\Phi$ on a set of trajectories will be called a  \emph{swapping mapping} if it operates by applying a sequence of swaps.

The main idea now will be to construct a  swapping mapping whose goal will be to transform trajectories of the algorithm to a  form  that satisfies certain properties   by applying swaps .

For a trajectory $\Sigma$  in which  a tree $\tau \in \mathcal{W}_i$ occurs, we denote by  $W^{\tau}_\Sigma$ the  prefix of $W(\Sigma)$ up to the step that corresponds to the root of $\tau$ (observe that Proposition~\ref{uniqueness} mandates that there exists a unique such step). Notice that, since $\tau \in \mathcal{W}_i$, the algorithm addresses flaw $f_i$ at this step, and thus the final element of  $W^{\tau}_{\Sigma}$ is $\{i\}$. Finally, recall the definitions of $\mathcal{R}_i^{\pi}$, $\chi_{\pi}$ and  $\chi_i^{\pi}$.  
\begin{lemma}\label{step:swapping}
Fix a witness tree $\tau \in \mathcal{W}_i$ and let $\mathcal{X}^{\tau}$ be a valid set of trajectories in which $\tau$ occurs. 
If \\ $\mathcal{A} = (F,\sim,\rho)$ is commutative then there exists a set of trajectories $\mathcal{X}^{\tau}_\pi$ and a swapping mapping $\Phi^{\tau}:\mathcal{X}^{\tau}\rightarrow \mathcal{X}^{\tau}_\pi$ which is a bijection
such that \\
(a) for any $\Sigma\in \mathcal{X}^{\tau}_\pi$  we have that $W^{\tau}_\Sigma $  is the unique witness sequence in $\mathcal{R}_i^{\pi}$   such that $\chi_i^{\pi}(W_{\Sigma}^{\tau}) =  \chi_{\pi} (\tau) $; \\
(b) for any witness sequence $W$ the set $\{\Sigma\in \mathcal{X}^{\tau}_\pi\:|\:\mathrm{Rev}[W^{\tau}_\Sigma]=W\}$ is valid.
\end{lemma}
We prove Lemma~\ref{step:swapping} in Section~\ref{main_proof}.
To see  Theorem~\ref{WitnessTreeLemma}, consider  a witness tree $\tau \in \mathcal{W}_i$, and let $\mathcal{Y}^{\tau}$ be the set of all  trajectories that $\mathcal{A}$ may follow in which $\tau$ occurs. Now remove from $\mathcal{Y}^{\tau}$ any trajectory $\Sigma$ for which there exists a trajectory $\Sigma'$ such that $\Sigma$ is a proper prefix of $\Sigma'$ to get $\mathcal{X}^{\tau}$. Clearly, this is a valid set and so  recalling that $\chi_{\pi}$ is a bijection and  applying Lemma~\ref{step:swapping}  we have that:
\begin{align}\label{prwto_vima}
\Pr[ \tau ]  = \sum_{ \Sigma \in \mathcal{X}^{\tau} }  \Pr [ \Sigma ]    \le       \sum_{ \Sigma \in \mathcal{X}^{\tau} }  p( \Sigma )           =  \sum_{ \Sigma \in \mathcal{X}_{\pi}^{\tau} }  p (\Sigma )  \enspace,
\end{align}
 where to get the second equality we use  the second requirement of Definition~\ref{commutativity}.  Lemma~\ref{step:swapping} further implies that for every trajectory $\Sigma \in \mathcal{X}_{\pi}^{\tau}$ we have that $W_{\Sigma}^{\tau}$ is the (unique) witness sequence in $\mathcal{R}_i^{\pi}$  such that $\chi_i^{\pi}(W_{\Sigma}^{\tau})  = \chi_{\pi}(\tau)$, i.e.,  $ W_{\Sigma}^{\tau}     =  ( \chi_i^{\pi} )^{-1} \left( \chi_{\pi}(\tau) \right)  $ . This means that the witnesses of the trajectories in $\mathcal{X}_{\pi}^{\tau}$ have $W:= ( \chi_i^{\pi} )^{-1} \left( \chi_{\pi}(\tau) \right) $ as a common prefix. Since  part $(b)$ of Lemma~\ref{step:swapping} implies that $\mathcal{X}_{\pi}^{\tau}$ is valid, applying  Lemma~\ref{step_prefixes} we get:
\begin{align}
 \sum_{\Sigma \in  \mathcal{X}_{\pi}^{\tau}   } p ( \Sigma )  \le \lambda_{\mathrm{init}} \prod_{w \in W} \gamma( f_{ w} ) =\lambda_{\mathrm{init}}  \prod_{v \in V(\tau)} \gamma( f_{ (v)} )   \enspace \enspace,
\end{align}
where the second inequality follows from the fact that $\chi_{i}^{\pi}( W) = \chi_{\pi}( \tau )$ and $V(\tau) = V( \chi_{\pi}(\tau) ) $, concluding the proof.

\subsection{Proof of Lemma~\ref{step:swapping}}\label{main_proof}
Our proof builds on  the proof of Theorem 19 in~\cite{Kolmofocs}.  We will be denoting witness sequences  $W = (w_1, w_2, \ldots, w_t)$ as a sequence of \emph{named indices of flaws} $W =  (\mathbf{w}_1, \ldots, \mathbf{w}_t)$ where $\mathbf{w}_j = (w_j,n_j)$ and $n_j = | \{ k \in [j] \mid w_k = w_j  \} | \ge 1$ is the number of occurrences of  $w_j$ in the length-$j$  prefix of $W$. Note that a named index $\mathbf{w}$ cannot appear twice in a sequence $W$. Finally, if $\mathbf{w} $ is a named index of flaw we denote by $w$ (that is,  without bold font) the flaw index that is associated with it. 

For a trajectory $\Sigma$ such that $W(\Sigma)  =  ( \mathbf{w}_1, \ldots, \mathbf{w}_t )$ we define a directed acyclic graph $\mathbf{G}(\Sigma) = ( \mathbf{V}(\Sigma), \mathbf{E}(\Sigma)) $ where $\mathbf{V}(\Sigma) = \{ \mathbf{w}_1, \ldots, \mathbf{w}_t \} $ and $\mathbf{E}(\Sigma) =  \{ \left( \mathbf{w}_j, \mathbf{w}_k \right)  \text{ s.t. }  w_j \sim w_k  \text{ and } j< k $  \}. This means that we have an edge from a named  flaw $\mathbf{w}_i$ to another flaw $\mathbf{w}_j$ whenever their corresponding flaw indices are related according to $\sim$ and  $\mathbf{w}_j$ occurs in $\Sigma$ before $\mathbf{w}_k$.

By Proposition~\ref{uniqueness}, for any trajectory $\Sigma$ in which $\tau$ occurs there is a unique step $t^* = t^*(\Sigma)$ such that $\tau_{W(\Sigma)}(t^*) = \tau$. For such a trajectory $\Sigma$, let $\mathbf{Q}(\Sigma) \subseteq \mathbf{V}(\Sigma)$ be the set of flaws from which the node $\mathbf{w}_{t^*}$ can be reached in $\mathbf{G}(\Sigma)$, where $\mathbf{w}_{t^*}$ is the named flaw index that corresponds to the step $t^*$. Notice that $w_{t^*} = i$ (since $\tau \in \mathcal{W}_i$). For $\mathbf{w} \in \mathbf{Q}(\Sigma)$ let $d(\mathbf{w} ) $ be the length of the longest path from $\mathbf{w}$ to $\mathbf{w}_{t^*}$ in $\mathbf{G}(\Sigma)$ plus one. For example, $d(\mathbf{w}_{t^*} ) = 1$.

Let $Q(\Sigma)$ denote the sequence consisting of the named flaws in $\mathbf{Q}(\Sigma)$ listed in the order they appear in $\Sigma$. The idea is to repeatedly apply the operation $\mathrm{Swap}$  to $\Sigma$ so that we reach a trajectory $\Sigma'$ that has a permutation $Q_{\pi}(\Sigma)$ of $Q(\Sigma)$  as a prefix. In particular, we will show that $Q_{\pi}(\Sigma) \in \mathcal{R}_i^{\pi}$ and $ \chi_{i}^{\pi}( Q_{\pi}(\Sigma) ) = \chi_{\pi}( \tau)$.

To that end,  for an integer $r\ge 1$ define $\mathbf{I}_r = \{ \mathbf{w } \in \mathbf{Q}(\Sigma) \mid d( \mathbf{w} ) = r  \}$, and let $Q_r$ be the sequence   consisting of the named flaw indices in $\mathbf{I}_r$ sorted in decreasing order with respect to $\pi$. Then, we define $Q_{\pi}( \Sigma)  =( Q_s ,\ldots, Q_1)  $ where $s = \max \{ d(\mathbf{w})  \mid \mathbf{w} \in \mathbf{Q}(\Sigma)  \} $.
 \begin{lemma}
 $Q_{\pi}(\Sigma) \in \mathcal{R}_i^{\pi}$ and $ \chi_{i}^{\pi}( Q_{\pi}(\Sigma) ) =  \chi_{\pi} (\tau) $.
\end{lemma}
\begin{proof}
Let $Y = Y(\Sigma) =  \mathrm{Rev}[Q_{\pi}(\Sigma)] = ( Q_1, Q_2, \ldots, Q_{s}) $ be the reverse sequence of $\Sigma$. By definition, $R_Y = Q_1 = \{i \}$.  To show that $Q \in \mathcal{R}_i^{\pi}$ it suffices to show  that $Q_{i+1} \subseteq \Gamma(Q_i)$ for each $i \in [s-1]$.  To see this, recall the definitions of  $ \mathbf{I}_{r+1}$ and $Q_{r+1}$ and observe that, for each $i_{r+1}  \in Q_{r+1}$, there must be a path of $r$ indices of flaws $i_r, i_{r-1}, \ldots, i_{1}  $ such that  for every $j \in [ r - 1]$ we have that $i_j \in Q_{j} $ and $i_j \sim i_{j+1} $.

Let $k$ be the number of elements in witness sequence $Q(\Sigma)$. Recall  that $\chi_{\pi}^i(Q_{\pi}(\Sigma)) := \chi_{\pi} (\tau_{Q_{\pi}(\Sigma) }(k))$  (proof of Lemma~\ref{trees_stable}). 
The proof is concluded by also recalling the algorithm for constructing witness trees and observing  that $ \tau_{Q_{\pi}(\Sigma)}(k) =  \tau_{W(\Sigma)}(t^*) = \tau $. 

\end{proof}
Note  that applying $\mathrm{Swap}$ to $\Sigma$ does not affect graph $\mathbf{G}(\Sigma)$ and set $\mathbf{Q}(\Sigma)$ and, thus, neither the sequence $Q_{\pi}(\Sigma)$. With that in mind, we show next how we could apply  $\mathrm{Swap}$ repeatedly to $\Sigma \in \mathcal{X}^{\tau}$ to reach a $\Sigma'$ such  $Q_{\pi}(\Sigma)$ is a prefix of its witness sequence (that is,  $W(\Sigma') = (Q_{\pi}(\Sigma), U )$). We will do this by applying swaps to \emph{swappable pairs} in $\Sigma$.
\begin{definition}
Consider a trajectory $\Sigma \in \mathcal{X}^{\tau}$. A pair $(\mathbf{w}, \mathbf{y} )$ of named indices of flaws is called \emph{a swappable} pair in $\Sigma$ if it can be swapped in $\Sigma$ (i.e., $W(\Sigma)  = (\ldots \mathbf{w},\mathbf{y} \ldots  )$ and $\mathbf{w} \nsim \mathbf{y}$) and either 

\begin{enumerate}

\item$ (\mathbf{w,y}) \in \left( \mathbf{V}(\Sigma) \setminus \mathbf{Q}(\Sigma)   \right)  \times   \mathbf{Q}(\Sigma)  $, or

\item $(\mathbf{w,y }) \in \mathbf{Q}(\Sigma) \times \mathbf{Q}(\Sigma)$ and their order in $Q_{\pi}(\Sigma)$ is different: $Q_{\pi}(\Sigma) =  (\ldots ,\mathbf{y } , \mathbf{w}, \ldots )$

\end{enumerate}

The position of the rightmost swappable pair in $\Sigma$ will be denoted as $k(\Sigma)$ , where the position of $(\mathbf{w}, \mathbf{y}   )$  in $\Sigma$ is the number of named indices of flaws that precede $\mathbf{y}$ in $W(\Sigma)$. If $\Sigma$ does not contain a swappable pair then $k(\Sigma) = 0 $. Thus, $k(\Sigma) \in [0, |\Sigma | -1]$.
\end{definition}

We can only apply finite many swaps to swappable pairs in $\Sigma$. This is because swapping pairs of the first form moves a named index in $\mathbf{Q}(\Sigma)$ to the left, while swapping pairs of the second one decrease the number of pairs  whose relative order in $Q(\Sigma)$ is not consistent with the one in $Q_{\pi}(\Sigma)$. Clearly, both of these actions can be performed only a finite number of times.

The following lemma shows how we can obtain a mapping $\Phi^{\tau}$ such that $ \mathcal{X}_{\pi}^{\tau} := \Phi^{\tau} (  \mathcal{X}^{\tau}) $ satisfies the first condition of Lemma~\ref{step:swapping}. The proof is identical (up to minor changes)  to the one of Lemma 27 of~\cite{Kolmofocs}. We include it in Section~\ref{Misc} for completeness.

\begin{lemma}\label{Kolmo27}
Consider a trajectory $\Sigma \in \mathcal{X}^{\tau} $ such that  $W(\Sigma) = ( A,U)$  where $A$ and $U$ are some sequences of indices of flaws, and there are no swappable pairs inside $U$. Then $U = (B,C)$ where sequence $B$ is a subsequence of $Q_{\pi}(\Sigma)$ and $C$ does not 
contain named indices of flaws from $\mathbf{Q}(\Sigma)$.

In particular, if $|A| = 0$ and $W(\Sigma) = U$ does not contain a swappable pair then $W(\Sigma)= (Q_{\pi}(\Sigma),C)$. 
\end{lemma}

It remains to show that  $\Phi^{\tau}$ can be constructed so that is also  a bijection and  that it satisfies the second condition of Lemma~\ref{step:swapping}. To do so, consider the following algorithm.

\begin{itemize}

\item Let $\mathcal{X}_0 = \mathcal{X}^{\tau}$.

\item While $k = \max_{\Sigma \in \mathcal{X}_p } k( \Sigma) \ne 0 $

\begin{itemize}

\item  For each $\Sigma \in \mathcal{X}_p$: if $k(\Sigma) = k$ then swap the pair $(\mathbf{w}, \mathbf{y } ) $ at position $k$ in $\Sigma$, otherwise leave $\Sigma$ unchanged.

\item Let $\mathcal{X}_{p+1} $ the new set of trajectories.

\end{itemize}

\end{itemize}

For a witnesses sequence $W$ define $\mathcal{X}_p[W] = \{ \Sigma \in \mathcal{X}_p \mid Q_{\pi}(\Sigma)  = W \}$ for an index $p \ge 0$. Now the following lemma concludes the proof since $\mathcal{X}_0[W] \subseteq  \mathcal{X}^{\tau} $ is valid.  Its proof is identical (up to minor changes) to the proof of Lemma 28
in~\cite{Kolmofocs}. We also include it in Section~\ref{Misc}  for completeness.

\begin{lemma}\label{Kolmo28}
If set $\mathcal{X}_p[W]$ is valid then so is $\mathcal{X}_{p+1}[W]$, and the mapping from $\mathcal{X}_p[W] $ to $\mathcal{X}_{p+1}[W]$ defined by the algorithm above is injective.
\end{lemma}

\section{Byproducts of Theorem~\ref{main}}\label{Byproducts}

In this section we show two important byproducts of Theorem~\ref{main} which we will use in our applications.

\subsection{Entropy of the Output Distribution}\label{EntropyOutput}

An important application of the known bounds for the Moser-Tardos distribution is estimating its randomness. In particular, Harris and Srinivasan~\cite{EnuHarris} show that  one can give lower bounds on the \emph{R\'{e}nyi entropy} of the output of the Moser-Tardos algorithm.

\begin{definition}[\cite{min_entropy3}]
Let $\nu$ be a probability measure over a finite set $S$. The \emph{R\'{e}nyi entropy with parameter $\rho$} of $\nu$ is defined to be
\begin{align*}
H_{\rho}[\nu]  = \frac{1}{1 - \rho} \ln \sum_{s \in S } \nu(s)^{\rho}  \enspace.  
\end{align*}
The \emph{min-entropy}  $H_{\infty}$ is  a special case defined as $H_{\infty}[\nu] = \lim_{\rho \to \infty} H_{\rho}[\nu] = - \ln \max_{ s \in S } \nu(s)$.  
\end{definition}

Using the results of Section~\ref{Results} we can show the analogous result in our setting.

\begin{theorem}\label{EntropyBound}
Assume that $\mathcal{A} = ( F, \sim ,\rho) $ is commutative, and the cluster expansion condition is satisfied. Let $\nu$ be the output distribution of $\mathcal{A}$. Then, for $\rho >1$,
\begin{align*}
H_{\rho}[\nu] \ge H_{\rho}[\mu]  - \frac{\rho}{\rho-1} \ln  \left(   \sum_{S \in \mathrm{Ind}([m])  } \prod_{j \in S}   \psi_j   \right) - \frac{\rho }{ \rho -1} \ln \lambda_{\mathrm{init} }     \enspace.
\end{align*}
\end{theorem}

 Given Theorem~\ref{main}, the proof is akin to the analogous result in~\cite{EnuHarris} and can be found in  Appendix~\ref{Misc} .

\begin{remark}\label{shearer_remark}
Using  Shearer's condition we can replace $\psi_i$ with $\frac{q_{ \{i \}} (\gamma)  } {q_{\emptyset}(\gamma) }  $, $i \in [m] $.
\end{remark}

A straightforward application of having a lower bound on  $H_{\rho}[\nu] $ (for any $\rho$), where $\nu$ is the output distribution of the algorithm, is that there exist at least $\mathrm{exp}( H_{\rho}[\nu]   ) $ flawless objects.  Before~\cite{EnuHarris}, the authors in~\cite{lu2009new} also used the (existential) LLL for enumeration of combinatorial structures by exploiting the fact that it guarantees a small probability $p$ of avoiding all flaws when sampling from the uniform measure (and, thus, their number is at least $p | \Omega |$).

\subsection{Partially Avoiding Flaws}\label{Partially}

One of the main results of~\cite{Haeupler_jacm,EnuHarris} are LLL conditions for the existence of objects that avoid  a large portion of the bad events. For instance, given a  sufficiently sparse $k$-SAT formula which, nonetheless,   violates the original LLL conditions, one can still find  an assignment that satisfies many clauses. Using the results of Section~\ref{Results}, we are able to extend the (most general) result of  Harris and Srinivasan~\cite{EnuHarris}  to the commutative setting.

Given a sequence of positive numbers $\{\psi_i \}_{i=1}^{m}$, for each $i \in [m]$ define: 
\begin{align*}
\zeta_i := \sum_{ S \in \mathrm{Ind}(\Gamma(i) )  }  \prod_{j \in S} \psi_j \enspace,
\end{align*} 
and notice that the cluster expansion condition can be expressed as requiring that for each $i \in [m]$ we have that $ \gamma(f_i)  \zeta_i \le \psi_i$. 

\begin{theorem}\label{partial_avoiding}
Assume that $\mathcal{A} = (  F, \sim ,\rho) $ is commutative and $\lambda_{\mathrm{init}} =1 $. Let $\{ \psi_i \}_{i=1}^{m}$ be a sequence of positive numbers. Then there is an algorithm $\mathcal{A}'$ (which is a modification of $\mathcal{A}$) and whose output distribution $\nu$ has the property that for each $i \in [m]$
\begin{align*}
\nu(f_i) \le \max \{ 0,  \gamma(f_i) \zeta_i  - \psi_i \}  \enspace.
\end{align*}
Furthermore, the expected number of times a flaw $f_i$ is addressed is  at most $  \psi_i$.
\end{theorem}

Given Theorem~\ref{main}, the proof of Theorem~\ref{partial_avoiding} is akin to the one of~\cite{EnuHarris}  and can be found in   Appendix~\ref{Misc}.

\begin{remark}\label{shearer_remark_2}
Using  Shearer's condition we can replace $\psi_i$ with $\frac{q_{ \{i \}} (\gamma)  } {q_{\emptyset}(\gamma) }  $, $i \in [m] $.
\end{remark}

\section{Applications}\label{Applicatia}

In this section we show concrete applications of our main results in several problems.

\subsection{Rainbow Matchings}\label{Rainbow_Matchings}

In an edge-colored graph $G=(V,E)$, say that $S \subseteq E$ is \emph{rainbow} if its elements have distinct colors. In this section we consider the problem of finding rainbow matchings in complete 
graphs of size $2n$, where each color appears a limited amount of times.

Applying the cluster expansion  condition, it can be shown~\cite{AIJACM,HV} that  any edge-coloring of a complete graph of size $2n$ in which each color appears on at most $\frac{27}{128} n \approx 0.211 n$ edges admits a rainbow perfect matching
that can be found efficiently. Furthermore, in~\cite{Kolmofocs} it is shown that the resampling oracles defined by~\cite{HV} for the space of matchings in a clique of even size, and which are used in this particular application, induce  commutative algorithms. The latter fact  will allow us to use our results to further study this problem. 

\subsubsection{Finding Rainbow Perfect Matchings}\label{algo}

We first formulate the problem to fit our setting and use Theorem~\ref{main} to show that the algorithm  of~\cite{AIJACM,HV} finds a perfect rainbow matching efficiently. Assuming a multiplicative slack in the cluster expansion conditions, a  running time (number of steps) of $O( n)$ can be given using the results of $\cite{AIJACM,HV,Kolmofocs}$. However, the best known upper bound without this assumption was given in~\cite{HV} to be $O( n^2)$. Here we improve the latter to $O(  n)$.  

Let $\phi$ be any edge-coloring of $K_{2n}$ in which each color appears on at most $\lambda n$ edges. Let $P=P(\phi)$ be the set of all pairs of vertex-disjoint edges with the same color in $\phi$, i.e., $P = \{ \{e_1,e_2 \}: \phi(e_1) = \phi(e_2) \}$. 
Let $\Omega$ be the set of all perfect matchings of $K_{2n}$. For each $\{e_i,e_j\} \in P$ let
\[
f_{i,j} = \{ M \in \Omega:  \{e_i,e_j\} \subset M\} \enspace .
\]
Thus, an element of $\Omega$ is flawless iff it is a rainbow perfect matching. The algorithm that finds a rainbow perfect matching starts at a state of $\Omega $ chosen uniformly at random and, in every subsequent step, it chooses (arbitrarily) a flaw to address. Algorithm~\ref{Rainbow_Matchings} below describes the probability distributions $\rho_{i,j}(M, \cdot) $, where $M \in f_{i,j} $. This a special case of the implementation of a general resampling oracle with respect to the uniform measure over $\Omega$ for perfect matchings  described in~\cite{HV}. For the problem of rainbow matchings, the latter implies that $\gamma(f_{i,j}) = \mu(f_{i,j} ) =  \frac{1}{ (2n-1 )(2n-3 )} $.

\begin{algorithm}[H]
\begin{algorithmic}[1]    

 \State $M':= M$, $A := \{e_1,e_2 \}$,  $A':= A$.

 \While{ $A' \ne \emptyset$ }

\State Pick $(u,v) \in A'$ arbitrarily

\State Pick $(x,y) \in M' \setminus A'$ uniformly at random, with $(x,y)$ randomly ordered;

\State With probability $1 - \frac{1}{ 2 | M' \setminus A'| +1} $,

Add $(u,y), (v,x)$ to $M'$ and remove $(u,v),(x,y)$ from $M'$;

\State Remove $(u,v)$ from $A'$;

 \EndWhile 
 
\State Output $M'$.

\caption{Probability Distribution $\rho_{i,j}(M, \cdot ) $}
\label{Rainbow_Matchings}
\end{algorithmic}
\end{algorithm}
For a vertex $v$  let $\Gamma(v)$ denote the set of  indices of flaws that correspond to edges adjacent to $v$. By observing the algorithm it's not hard   to verify (and is also proved in~\cite{AIJACM,HV,Kolmofocs}) that  the graph $C$ over indices of flaws such that for each $\left(e_i = (v_1,v_2),e_j = (v_3,v_4) \right) \in P$ we have that
\begin{align*}
\Gamma \left( i,j \right) =  \bigcup_{i=k}^{4} \Gamma(v_k) 
\end{align*}
is a causality graph. Furthermore, if $S \in \mathrm{Ind}\left( \Gamma  (i, j)  \right)$, then for each $k \in \{1,2,3,4 \}$ we have that $| S \cap \Gamma(v_k) | \le 1$. This means that $|S| \le 4$ and, moreover,  for each $j \in \{0,1,2,3,4\}$ there are at most
${4 \choose j } (2n-1)^j (\lambda n - 1)^j$ subsets $S \in \mathrm{Ind}( \Gamma(i,j) ) $ of size $j$. Choosing parameters $\psi_{i,j} = \psi  = \frac{3 }{4n^2 } $ we have that:
\begin{align*}
\gamma(f_{i,j} ) \zeta_{i,j} :=  \gamma(f_{i,j} ) \sum_{ S \in \mathrm{Ind}(\Gamma(i,j)) }  \psi^{|S|}  \le \frac{ 1}{ (2n-3)(2n-1)} \left( 1 + (2n-1)(\lambda n -1)  \right)^4 \enspace,
\end{align*}
from which it can be seen that whenever $\lambda \le \frac{27}{128} $ we have that  $ \gamma(f_{i,j} ) \zeta_{i,j}  \le 1$ and so the cluster expansion condition is satisfied.

Since $|P | \le (2n)^2   \cdot ( \lambda n -1)  <  4 \lambda n^3 $,  Theorem~\ref{main} implies that the algorithm terminates after an expected number of  $ 3 \lambda n  $ steps. Overall, we have showed the following theorem.
\begin{theorem}
For any  $\lambda \le \frac{27}{128} $, given any edge-coloring of the complete graph on $2n$ vertices in which each color appears on at most $\lambda n$ edges, there exists an algorithm that terminates in an expected number of at most $3 \lambda n$ steps and outputs a rainbow perfect matching.
\end{theorem}

\subsubsection{Number of Rainbow Perfect  Matchings}\label{number_rainbow_matchings}

In this subsection  we use Theorem~\ref{EntropyBound} to give an exponential lower bound on the  number of perfect matchings when each color appears at most $\lambda n$ times, where $ \lambda \le \frac{27}{128}$, by bounding the entropy of the output distribution of the algorithm described in the previous subsection.
\begin{theorem}\label{Matchings_Enu}
For any  $\lambda \le \frac{27}{128} $, given any edge-coloring of the complete graph on $2n$ vertices in which each color appears on at most $\lambda n$ edges, there exist at least \footnote{Recall that $(2n- 1)!!=  1\cdot 3 \cdot \ldots  \cdot (2n-1) = \frac{(2n)! }{2^n n ! }  $ .} $e^{-3 \lambda n }\cdot (2n-1)!! $ rainbow perfect matchings. Furthermore, there exists an algorithm that outputs each one of them with positive probability.
\end{theorem}
\begin{proof}

To apply Theorem~\ref{EntropyBound}, we will need to give an upper bound for  $\sum_{  S \in \mathrm{Ind}([m]) } \psi^{|S|} $. Similarly to applications in~\cite{EnuHarris}, we will find useful the following crude, but general upper bound:
\begin{align}\label{crude}
\sum_{  S \in \mathrm{Ind}([m]) }\prod_{j \in S} \psi_j  \le  \sum_{ S \subseteq [m]}   \prod_{j \in S} \psi_j   \le   \prod_{i \in [m] } (1 + \psi_i ) \le  \mathrm{exp}{  \left( \sum_{i \in [m]} \psi_i \right)  }  \enspace.
\end{align}
Since the number of perfect matching in $K_{2n}$ is $(2n-1)!!$ and   also $|P | <   4 \lambda n^3 $, Theorem~\ref{EntropyBound} and~\eqref{crude} imply  that  the number of  rainbow perfect matchings is at least 
\begin{align*}
\mathrm{exp} \left(  \ln | \Omega | - \sum_{ i \in [m]  }   \psi   \right) \ge \mathrm{exp} \left(  \ln \left(  (2n-1)!! \right)  -  3 \lambda n \right)  =  \frac{ (2n-1)!!}{e^{3 \lambda n} }   \enspace,
\end{align*}
concluding the proof.
\end{proof}

\subsubsection{Low Weight Rainbow Perfect Matchings}\label{LowWeightMatchings}
Consider an arbitrary weighting function $W: E \rightarrow \mathbb{R}$ over the edges of $K_{2n}$. Here we consider the problem of finding rainbow perfect matchings of low weight, where the weight of a matching is defined
as the sum of weights of its edges. Clearly, there is a selection of $n$ edges of $K_{2n}$ whose total weight is at most $\frac{1/2}{2n-1}\sum_{ e \in K_{2n}} W(e)$.  We use Theorem~\ref{main} to show that, whenever $\lambda \le \frac{27}{128}$, the algorithm of subsection~\ref{algo} outputs a rainbow perfect matching of similar expected weight. 
\begin{theorem}\label{Matchings_Weight}
For any  $\lambda \le \frac{27}{128} $, given any edge-coloring of the complete graph on $2n$ vertices in which each color appears on at most $\lambda n$ edges, there exists an algorithm that outputs a perfect rainbow matching $M$ such that
\begin{align*}
\ex[W(M)] \le  \frac{ \left(  1 + \frac{3}{2} \lambda   \right)^2 }{2n-1}  \sum_{e \in K_{2n} } W(e)  \enspace.
\end{align*}
\end{theorem}
\begin{proof}
Let  $A_e$ be the subset of $\Omega$ that consists of the matchings that contain $e$. It is proven in~\cite{HV}, and it's also not hard to verify, that Algorithm~\ref{Rainbow_Matchings} with $A = \{ e\}$  is a resampling oracle  for this type of flaw. Moreover, using an identical counting  argument to the one in subsection~\ref{algo} we get that:
\begin{align*}
\sum_{ S \in  \mathrm{Ind}(\Gamma(A_e))  } \psi^{|S|}  \le \left(  1 + (2n-1)(\lambda n -1) \psi \right)^2 \enspace.
\end{align*}
Applying Theorem~\ref{main} we get that:
\begin{eqnarray*}
\ex[W(M)] &\le &  \sum_{e \in K_{2n} } W(e) \Pr[ A_e ] \\
		& \le & \sum_{e \in K_{2n} } W(e) \mu \left(A_e \right) \left(  1 + (2n-1)(\lambda n -1) \psi \right)^2 \\
		 & < & \frac{ \left(  1 + \frac{3}{2} \lambda  \right)^2 }{2n-1}  \sum_{e \in K_{2n} } W(e)  \enspace,
\end{eqnarray*}
concluding the proof.
\end{proof}

\subsubsection{Finding Rainbow  Matchings with many edges}\label{rainbow_matchings}
 In this subsection we use  Theorem~\ref{partial_avoiding}  to show that whenever $\lambda < 0.5$ we can find rainbow matchings with a linear number of edges.
\begin{theorem}\label{Matchings_Partial}
 Given any edge-coloring of the complete graph on $2n$ vertices in which each color appears on at most $\lambda n$ edges, where  $\lambda < 0.5 $    and $n$ is sufficiently large, there exists an algorithm that terminates within $O( n) $ steps in expectation and finds a rainbow matching with an expected number of edges that is at least $n \min \left( 1, 0.94  \sqrt [3] {  \frac{2}{\lambda} } -1 \right)$.
\end{theorem}
\begin{proof}
Let $\phi$ be any edge-coloring of $K_{2n}$ in which each color appears on at most $\lambda n$ edges and recall the definitions of $P=P(\phi)$,  $\Omega$, and  $f_{i,j} = \{ M \in \Omega:  \{e_i,e_j\} \subset M\}$ from the proof of Theorem~\ref{Matchings_Enu}.

The idea is to apply Theorem~\ref{partial_avoiding}  that guarantees that we can come up with a ``truncated version" $\mathcal{A}'$ of our algorithm for finding perfect rainbow matchings. In particular, if $\nu$ is the output probability distribution of $\mathcal{A}'$ and for each flaw $f_{i,j}$ we set $\psi_{i,j} =  \alpha$ then:
\begin{eqnarray}
\nu(f_{i,j} )  & \le &  \max \left(0, \gamma(f_{i,j}) \zeta_{i,j}  - \alpha  \right) \nonumber \\
		 & \le  &\max \left(0,   \frac{  \left( 1 + (2n-1) (\lambda n-1) \alpha  \right)^4  }{ (2n-3) (2n-1) }  - \alpha \right) \label{some_flaws} \enspace .
\end{eqnarray}
Consider now the following strategy:  We first execute algorithm $\mathcal{A}'$ to get a perfect, possibly non-rainbow,  matching $M$ of $K_{2n}$. Then, for each flaw $f_{i,j}$ that appears in $M$, we choose arbitrarily one of its corresponding edges and remove it from $M$, to get a non-perfect, but rainbow, matching $M'$. If $ S = S(M')$ is the random variable that equals the size (number of edges) of $M'$ then by  setting  
$$\alpha =  \frac{1}{(2n-1) (\lambda n-1) }  \left( \sqrt[3]{ \frac{2n-3}{4( \lambda n-1)} }- 1\right)  $$
 we get:
\begin{eqnarray*}
\ex[S] & =  & n - \sum_{ (e_i, e_j) \in P } \nu(f_{i,j} )  \\
	  & \ge & n  -  \max \left(0,     |P| \left(  \frac{  \left( 1 + (2n-1) (\lambda n-1) \alpha  \right)^4  }{ (2n-3) (2n-1) }  - \alpha \right) \right)  \\
	  & =  & n \min \left( 1,  1-   \frac{ 4 n  }{2n-1   } \left(     1  -   \frac{3 }{ 4  \cdot 2^{2/3} }   \sqrt [3] { \frac{2n-3 }{\lambda n -1 } }   \right)  \right)  \enspace.
\end{eqnarray*}
For large enough $n$, the latter is 	  $ \min \left( 1, 0.94  \sqrt [3] {  \frac{2}{\lambda} } -1 \right) $. Finally, notice that  (for large $n$)  $\alpha  $ is positive whenever $\lambda < 0.5$.
\end{proof}

\subsection{List-coloring Triangle-Free Graphs}\label{MikeResult}

In the problem of list-coloring one is given a graph $G = G(V,E) $ over $n$ vertices $V = \{v_1, v_2, \ldots, v_n \}$ and,  for each $v \in V$, a list of colors $\mathcal{L}_v$. The goal is to find a list-coloring $\sigma \in \mathcal{L}_{v_1} \times \ldots \times \mathcal{L}_{v_n}  $ of $G$ such that $\sigma(v) \ne \sigma(u)$ for any pair of adjacent vertices.

The \emph{list chromatic number} $\chi_{\ell}(G)$ of a graph $G$ is the minimum number of colors for which such a coloring is attainable. A celebrated result of Johansson shows that there exist a large constant $C > 0$ such that every triangle-free graph with maximum degree $\Delta \ge \Delta_0 $ can be list-colored using $ C \Delta / \ln \Delta $ colors.  Very recently, Molloy~\cite{molloy2017list} improved Johansson's result showing that $C$ can be replaced by $(1+\epsilon)$ for any $\epsilon > 0$ assuming that $\Delta \ge \Delta_{\epsilon}$.(We note that, soon after, Bernshteyn~\cite{MolloyLLL} established the same bound for the list chromatic number using the LLL. However, his result is not constructive as it uses a sophisticated probability measure for which it is not clear how one could design ``efficient" resampling oracles.)

Here we show how that the algorithm of  Molloy is amenable to our analysis and, in particular,  we prove that it can output exponentially many  proper colorings with positive probability.

\subsubsection{The Algorithm}

The algorithm of~\cite{molloy2017list}  works in two stages. First, it finds a partial list-coloring   which has the property that (i) each vertex $v \in V$ has ``many" available colors; (ii) there is not ``too much competition" for the available colors of $v$, i.e., they do not appear in the list of available colors of its neighbors. Then, it completes the coloring via a fairly straightforward  application of the Moser-Tardos algorithm. 

To describe the algorithm formally, we will need some further notation. First, it will be convenient to treat $\mathrm{Blank}$ as a color that is in the list of every vertex. For each vertex $v$ and \emph{partial} proper coloring $\sigma$  let 
\begin{itemize}
\item $N_v$ denote the set of vertices adjacent to $v$;
\item $L_v(\sigma) \subseteq \mathcal{L}_v$ to be the set of \emph{available} colors for  $v$ at state $\sigma$, i.e., the set of colors that we can assign to $v$ in $\sigma$ without making any edge monochromatic. Notice that $\mathrm{Blank}$ is always an available color;
\item $T_{v,c}(\sigma) $ to be the set of vertices $u \in N_v$ such that $\sigma(u) = \mathrm{Blank}$ and $ c \in L_u(\sigma)$.
 \end{itemize}
Let $\Omega =   \prod_{ v \in V} \mathcal{L}_v$ and define $L =    \Delta^{ \frac{ \epsilon } {2  }}   $. Given a partial list-coloring, we define the following flaws for any vertex $v$:
\begin{eqnarray*}
B_v &= & \left\{  \sigma \in \Omega:  | L_v(\sigma) | < L \right\}; \\
Z_v & =&  \left\{  \sigma  \in \Omega:   \sum_{  c \in L_v(\sigma) \setminus \mathrm{Blank} } | T_{v,c}(\sigma) | > \frac{1}{10 } L \cdot | L_v(\sigma)|  \right\} \enspace.
\end{eqnarray*}

\begin{lemma}[The Second Phase]\label{sufficient_proper_coloring}
Given a flawless partial list-coloring, a complete list-coloring of $G$ can be found in expected polynomial time.
\end{lemma}

Lemma~\ref{sufficient_proper_coloring} was proved in~\cite{molloy2017list} via a fairly straightforward application of the Lov\'{a}sz Local Lemma, and can be made constructive via the Moser-Tardos algorithm. We also present its proof in Appendix~\ref{second_phase}, as it will be useful in our analysis. What is left is to describe the first phase of the algorithm.

\begin{itemize}

\item The initial distribution $\theta$, which is important in this case, is chosen to be the following: Fix  an independent set $S$ of $G$ of size at least $n/(\Delta+1)$. (This is trivial to find efficiently via a greedy algorithm). Choose one color from $\mathcal{L}_u \setminus \mathrm{Blank}$, $u \in S$, uniformly at random, and assign it to $u$;

\item  to address a flaw $f \in \{B_v, Z_v \}$ at state $\sigma$, for each $u \in N_v$,  choose uniformly at random a color from $L_u(\sigma)$ and assign it to $u$;

\item as a flaw choice strategy, the algorithm first fixes any ordering $\pi$ over flaws. At every step, it chooses  the lowest occurring flaw according to $\pi$ and address it.

\end{itemize}

%

\subsubsection{Proving Termination}

Let $\mathcal{A}_1$, $\mathcal{A}_2$  denote the first and second phase of our algorithm, respectively.  Here we prove that $\mathcal{A}_1$ terminates in expected polynomial time. To do so, we will use the convergence result corresponding to equation~\eqref{AlgoLLL}. (Although, as  we will see, $\mathcal{A}_1$ is commutative for an appropriate choice of a causality graph, we won't  use Theorem~\ref{main} to prove its convergence. This is because $\lambda_{\mathrm{init} } $ is exponentially large in this case). 

The measure $\mu$ we use for the analysis is the  uniform measure over partial proper colorings. We will use the following lemma whose proof can be found in Section~\ref{key_proof}.
\begin{lemma}\label{key_bounds2}  For each vertex $v$  and flaw $f \in \{B_v, Z_v \} $ we have that 
\begin{align*}
\gamma(f) \le 2 \Delta^{-4} \enspace.
\end{align*}
\end{lemma}

Consider the causality graph such that  $f_v \sim f_u $,  if $\mathrm{dist}(u,v) \le 3$, where $f_v, f_u$ are either $B$-flaws or $Z$-flaws.  Notice that it has maximum degree at most $2( \Delta^3 +1)$.
Setting  $\psi_f = \psi =  \frac{1 }{ 2(\Delta^{3} +1) } $ for every flaw $f \in \{B_v,Z_v \}$ and applying~\eqref{AlgoLLL}, we get that the algorithm converges in expected polynomial time since
\begin{align*}
\gamma(f ) \sum_{ S \subseteq  \Gamma(f)   } \prod_{ g \in S } \psi_g    \le \frac{2}{ \Delta^4} \cdot  2(\Delta^3 +1) \cdot \mathrm{e}  < 4 \mathrm{e} \left(  \frac{1}{ \Delta}  + \frac{1}{\Delta^4} \right) < 1 \enspace,
\end{align*}
 for large enough $\Delta$,  and $ \log_2 | \Omega | + 2n \log_2 \left(  1 + \psi   \right)  = O (n \log n  ). $


\subsection{A Lower Bound on the Number of Possible Outputs}

In this section it will be convenient to  assume that the list of each vertex has size exactly $q$.

Let $\mathcal{A}_1$, $\mathcal{A}_2$  denote the first and second phase of our algorithm, respectively.  The bound regarding the number of list-colorings the algorithm can output with positive probability follows almost immediately from the two following lemmas.

\begin{lemma}\label{first_phase_many}
Algorithm $\mathcal{A}_1$ can output at least  $\mathrm{exp} \left( n \left( \frac{\ln q }{ \Delta+1 }  - \frac{1}{\Delta^3}  \right) \right)$ flawless partial colorings with positive probability.
\end{lemma}
\begin{proof}
It is not hard to verify that algorithm $\mathcal{A}_1$ is commutative with respect to the causality relation $\sim$ induced by neighborhoods $\Gamma(\cdot)$. To see this, notice that for any two flaws $f_v, f_u$ and any $\sigma \in f_v \cap f_u$, invoking procedure { \sc  Resample$(v,\sigma)$ } does not change the list of available colors of the neighbors of $u$. Applying Theorem~\ref{EntropyBound} (using the crude bound we saw in~\eqref{crude}) we get that the algorithm can output at least
\begin{align}\label{many_flawless}
\mathrm{exp} \left( \ln  \frac{| \Omega|}{ \lambda_{\mathrm{init}} }  - \sum_{ f \in F } \psi_f   \right)  = \mathrm{exp} \left( \ln  \frac{1}{ \max_{\sigma  \in \Omega} \theta(\sigma) } - \frac{2n}{ 2( \Delta^3 +1)}   \right) > \mathrm{exp} \left( n \left( \frac{\ln q }{ \Delta+1 }  - \frac{1}{\Delta^3}  \right) \right) \enspace
\end{align}
flawless partial colorings. 
\end{proof}

\begin{lemma}\label{combo}
Suppose $\mathcal{A}_1$ can output $N$ flawless partial colorings with positive probability. Suppose further that among these partial colorings, the ones with the lowest number of colored vertices have exactly $\alpha n$ vertices colored, where $\alpha \in (0,1)$. Then, $\mathcal{A}_2$ can  output at least  $\max\left( N q^{-(1-\alpha)n  } , \left(\frac{8L }{11 } \right)^{ (1-\alpha)n} \right) $  list-colorings with positive probability. 
\end{lemma}

The proof of Lemma~\ref{combo} can be found in Appendix~\ref{applicatia_appendix}.

\begin{proof}[Proof of Theorem~\ref{many_solutions}]

According to Lemma~\ref{first_phase_many} algorithm $\mathcal{A}_1$ can output at least  $N: =   \mathrm{exp} \left( n \left( \frac{\ln q }{ \Delta+1 }  - \frac{1}{\Delta^3}  \right) \right)$ flawless partial colorings. Moreover, according to Lemma~\ref{combo}, algorithm $\mathcal{A}_2$ can output at least 
\begin{align*}
\min_{ \alpha \in (0,1) } \left\{  \max\left( N q^{-(1-\alpha)n  } , \left(\frac{8L }{11 } \right)^{ (1-\alpha)n} \right)   \right\} 
\end{align*}
distinct full-list colorings.  Since  $N q^{-(1-\alpha)n} $ ,  $\left(\frac{8L }{11 } \right)^{ (1-\alpha)n}   $ are increasing and decreasing as functions of  $\alpha$, respectively, the value of $\alpha$ that minimizes our lower bound is the one that makes them equal, which can be seen to be
\begin{align*}
\alpha^*  := 1 - \frac{ \ln N }{ n \ln ( \frac{  8 L q}{11}   )  }   =  1 - \frac{  \frac{\ln q}{ \Delta+1} - \frac{1}{ \Delta^3 } }{ \ln q  + \ln ( \frac{8L}{11 })  }   \enspace.
\end{align*}
Therefore, algorithm $\mathcal{A}_2$ can output at least
\begin{align*}
\mathrm{exp} \left( n \left( \frac{\ln q }{ \Delta+1 }  - \frac{1}{\Delta^3}  \right) \right) \cdot q^{ -(1-\alpha^*) n }   =  \mathrm{exp}\left(  n \left(\frac{q}{ \Delta+1} - \frac{1}{\Delta^3}  \right)  (1 - \delta)    \right) \enspace,
\end{align*}
list-colorings, where $\delta := \frac{1}{ 1 + \frac{ \ln ( 8L/11) }{\ln q }  }    \in (0,1)$, concluding the proof.

\end{proof}

\subsubsection{Proof of Lemma~\ref{key_bounds2}}\label{key_proof}

It will be convenient to extend the notion of ``addressing a flaw  $f$ in a state $\sigma$" to  arbitrary states $\sigma \in \Omega $,  meaning that we recolor the vertices associated  with $f$ in the same way we would do it if the constraint corresponding to $f$ was indeed violated. Consider the following random  experiments.
\begin{itemize}
\item Address $B_v$  at an arbitrary state $\sigma \in \Omega $ to get a state $\sigma'$. Let $\Pr_{\sigma}[B_v]$ denote the probability that $\sigma' \in B_v$.
\item   Address $Z_v$  at an arbitrary state $\sigma \in \Omega $ to get a state $\sigma'$. Let $\Pr_{\sigma}[Z_v]$ denote the probability that $\sigma' \in Z_v$. 
\end{itemize}
Our claim now is  that
\begin{eqnarray}
\gamma(B_v)   & \le   &\max_{ \sigma' \in \Omega } \Pr_{\sigma'} [B_v] \label{random_bound1} \enspace ; \\
\gamma(Z_v)    &  \le &     \max_{ \sigma' \in \Omega } \Pr_{\sigma'}[Z_v]  \label{random_bound2}\enspace.
\end{eqnarray}
To see this, let $f_v \in \{B_v, Z_v \}$ and  observe that 
\begin{align*}
\gamma(f_v) =  \max_{\sigma' \in \Omega} \sum_{ \sigma \in f_v} \frac{\mu(\sigma)}{ \mu(\sigma') } \rho_v(\sigma,\sigma')  = \max_{\sigma' \in \Omega} \sum_{ \sigma \in \mathrm{In}_{f_v}(\sigma')  }   \frac{ 1}{ | \Lambda(\sigma) | } \enspace, 
\end{align*}
where $\Lambda(\sigma): = \prod_{u \in N_v} L_u(\sigma) $ is the cartesian product of the lists of available colors of each vertex $u \in N_v$ at state $\sigma$ and $\mathrm{In}_{f_v}(\sigma')$ is the set of states  $ \sigma  \in f_v$  such that $\sigma' \in A(f_v,\sigma)$. 

The   key observation now is that $\Lambda(\sigma') = \Lambda(\sigma) $ for each state $\sigma \in \mathrm{In}_{f_v}(\sigma')$. This is because any  transition of the form $\sigma \xrightarrow{f_v} \sigma'$ does not alter the lists of available colors of vertices $u \in N_v$, since the graph is triangle-free.
Thus, 
\begin{align*}
\gamma(f_v) = \max_{\sigma' \in \Omega} \frac{ |\mathrm{In}_{f_v}(\sigma') |}{ \Lambda(\sigma') }  = \max_{\sigma' \in \Omega} \Pr_{\sigma'}[ f_v ]\enspace,
\end{align*}
where the  second equality follows from the fact that there is a bijection between $ \mathrm{In}_{f_v}(\sigma') $  and the set of color assignments from $\Lambda(\sigma')$  to the vertices of $N_v$ that violate the constraint related to flaw $f_v$.

The following lemma concludes the proof.
\begin{lemma}[\cite{molloy2017list}]\label{trelo}
For every state $\sigma \in \Omega$ it holds that
\begin{enumerate}[(a)]
\item $\Pr_{\sigma}[ B_v] < \Delta^{-4}$;
\item $\Pr_{\sigma}[ Z_v] < \Delta^{-4}$.
\end{enumerate}
\end{lemma}

\subsection{Acyclic Edge Coloring via the Clique Lov\'{a}sz Local Lemma}\label{AECARAMOU}

An edge-coloring of a graph is \emph{proper} if all edges incident to each vertex have distinct colors. A proper edge coloring is \emph{acyclic} if it has no bichromatic cycles, i.e., no cycle receives exactly two (alternating) colors. The smallest number of colors for which a graph $G$ has an acyclic edge-coloring  is denoted by $\chi'_a(G)$.

Acyclic Edge Coloring (AEC), was originally motivated by the work of Coleman et al.~\cite{coleman1,coleman2} on the efficient computation of Hessians and, since then, there has been a series of works~\cite{NogaLLL,mike_stoc,Ndreca,Haeupler_jacm,CliqueLLL,acyclic} that upper bound $\chi'_a(G)$ for graphs with bounded degree.  The currenty best result was given recently by Giotis et al.\  in~\cite{Kirousis} who showed that  $ \chi'_a(G) \le 3.74 \Delta$ in graphs with maximum degree $\Delta$.

The analysis of~\cite{Kirousis}, while inspired by the algorithmic LLL, uses a custom argument that does not correspond to any of its known versions. Furthermore, their algorithm does not correspond to  an instantiation of the
Moser Tardos algorithm and does not seem to be commutative (assuming the natural formulation in our setting) and, thus, it's not amenable to our analysis. 

On the other hand, Kolipaka, Szegedy and Yixin Xu show in~\cite{CliqueLLL} that  $8.6 ( \Delta -1)$ colors  suffice   for the Moser Tardos algorithm  to converge in this setting. They do this by introducing the \emph{Clique  the Lov\'{a}sz Local Lemma}, a condition that is typically stronger than (although, technically, incomparable to)   the cluster expansion condition, but  weaker than the Shearer's condition. In fact, the Clique Lov\'{a}sz Local Lemma is a member of a ``hierarchy" of LLL conditions that are increasingly complex and  use an increasing amount of information about the structure of the dependency graph. On the limit, they give the Shearer's condition.

While the use of the Clique Lov\'{a}sz Local Lemma (or any other condition in the hierarchy of~\cite{CliqueLLL}) makes the results of~\cite{Haeupler_jacm} inapplicable,  it does allow us to use Theorems~\ref{main} and~\ref{EntropyBound}  which capture the cases where the Shearer's condition is satisfied.

 We show two results. Our first theorem says that the Moser Tardos algorithm applied on the acyclic edge coloring problem converges in polynomial time and has high output entropy whenever $q \ge 8.6 (\Delta-1)$.

\begin{theorem}\label{AECARES}
Given a graph $G=(V,E)$ with maximum degree $\Delta$ and $q \ge 8.6 (\Delta -1)$ colors, there exist at least $(\frac{q}{4})^{|E|}$  acyclic edge colorings of $G$. Furthermore, there exists an algorithm with expected polynomial running time that
outputs each one of them with positive probability.
\end{theorem}

Our second theorem considers a problem of weighted acyclic edge colorings. In particular, given a graph $G(V,E)$ let $W = \sum_{v \in V} W_v $ be a weighting function over edge $q$-colorings of $G$ such that each $W_v $, $v \in V$, is a function of the colors of the edges adjacent to $v$.  By sampling uniformly at random, one can find  an edge coloring $\phi$ of weight $\ex_{\phi \sim \mu } [ W(\phi) ] $, where $\mu$ is the uniform distribution over the edge $q$-colorings of $G$. Using Theorem~\ref{main}, we  show that whenever $q \ge 8.6 (\Delta-1) $ we can use the Moser-Tardos algorithm to find an acyclic edge coloring of similar weight (assuming that $\Delta$ is constant). 

\begin{theorem}\label{LowWeight}
Given a graph $G(V,E)$ with maximum degree $\Delta$ and $ q \ge 8.6(\Delta-1)$ colors, there exist an algorithm with expected polynomial running time that outputs an acyclic edge coloring $\phi_{\mathrm{out} }$ of expected weight at most
\begin{align*}
 \ex[\phi_{\mathrm{out} } ]  <  1.3^\Delta \cdot  \ex_{\phi \sim \mu } [ W(\phi)]  \enspace.
 \end{align*}
 \end{theorem}
The proofs of Theorems~\ref{AECARES},~\ref{LowWeight}, as well as more details on the Clique Lov\'{a}sz Local Lemma, can be found in Section~\ref{applicatia_appendix}.

\section{Acknowledgements}
The author is grateful to Dimitris Achlioptas and Alistair Sinclair  for  detailed comments and feedback, as well as to anonymous reviewers for comments and remarks.

\bibliographystyle{plain}
\bibliography{kolmo}

\appendix

 \newpage

\section{Dealing with Super-Polynomially Many Flaws}\label{super-poly}

In this section we discuss how one can deal with problems where the number of flaws is super-polynomial in the natural size of the problem using commutative algorithms.

In such a setting,  there are two issues to be resolved. The first issue  is that one should be able to show that the expected number of steps until convergence is polynomial, and thus, much less than $\Theta(|F|) $. The second issue is that one should have an efficient procedure for  finding a flaw that is present in the current state, or decide that no such flaw exists.

\paragraph{Polynomial-Time Convergence.}

As far as the  issue of polynomial-time convergence is concerned, there are at least three approaches one can follow. 

A first approach is to start the algorithm at a state $\sigma_1$ in  which the set  of flaws present   is of polynomial size, and then employ the main results from~\cite{AIJACM,Harmonic,Kolmofocs} which  guarantee that the algorithm will terminate after $O\left( |U(\sigma_1)| +  \max_{\sigma \in \Omega} \log_2 \frac{1}{  \mu(\sigma) } \right) $ steps with high probability. This approach does not require  the algorithm to be commutative, but it does require that the LLL condition is satisfied with a slack in order to establish quick termination.

 A second approach, which was first applied in the context of  the Moser-Tardos algorithm by Haeupler, Saha and Srinivsan~\cite{Haeupler_jacm}, is to find a \emph{core} set of flaws of polynomial size and apply a modified version of the algorithm that effectively ignores any non-core flaw. The hope  is that non-core flaws will never occur during the execution of this modified algorithm. Extended to our setting, 
one uses the following theorem which is a straightforward corollary of Theorem~\ref{main}.
\begin{theorem}\label{core_events}
Assume that $\mathcal{A} = (  F, \sim ,\rho) $ is commutative. Let $ I \subseteq [m] $ be a set of indices that corresponds to a core subset of $F$ and  assume there exist positive real numbers  $\{ \psi_i \}_{i=1}^{m}$ such that for every $i \in [m]$
\begin{align*}
\gamma(f_i)    \sum \limits_{S \in \mathrm{Ind} \left(\Gamma(i)\cap I \right)  }  \prod_{j \in S} \psi_j  \le  \psi_i \enspace.
\end{align*}
Then there exists a modification of $\mathcal{A}$ that terminates in an expected number of $ O\left( \lambda_{\mathrm{init}} \sum_{i \in I} \psi_i \right) $ steps and outputs a flawless element with probability at least $1 -  \sum_{ i \in [m] \setminus I } \lambda_{\mathrm{init} }\gamma(f_i) \sum_{ S \in  \mathrm{Ind}(\Gamma(i)  \cap I) } \prod_{j \in S} \psi_j$.
\end{theorem}

Finally, a third approach is to show that the causality graph can be decomposed into a set of cliques of polynomial size and then apply a result   of~\cite{Haeupler_jacm} which states  that, in this case, the running time of the algorithm is polynomial (roughly quadratic) in the size of the decomposition. To be more precise, we note that in~\cite{Haeupler_jacm} the latter result  is shown for the Moser-Tardos algorithm  in the variable setting and assuming the General LLL condition~\eqref{eq:LLL}, where the clique decomposition considered is induced by the random variables that form the probability space (one clique per variable). However, the proof for the general case is identical. Using Theorem~\ref{main} and recalling Remark~\ref{general} we can extend this result to our setting to get the following theorem.

\begin{theorem}\label{clique_decomp}
Let $\mathcal{A} = (D,F, \sim,\rho)$ be a commutative algorithm such that the causality graph induced by $\sim$ can be partitioned into $n$ cliques, with potentially further edges between them. Assume there exist  real numbers $\{ x_i\}_{i=1}^{m}$ in $(0,1)$ such that for  every $i \in [m]$ we have that
\begin{align*}
\gamma(f_i) \le x_i \prod_{ j \in \Gamma(i)} (1- x_j) \enspace,
\end{align*}
and let $\delta := \min_{ i \in [m] } x_i \prod_{j \in \Gamma(i) } (1 - x_j) $. Then, the expected number of steps performed by  $\mathcal{A}$ is at most $t = O \left( \lambda_{\mathrm{init} } \cdot \frac{n}{\epsilon}  \log \frac{ n \log ( 1/ \delta) }{ \epsilon} \right)$, and for any parameter $\eta$, $\mathcal{A}$ terminates 
within $\eta t$ resamplings with probability $1 - \mathrm{e}^{-\eta }$.
\end{theorem}
\begin{remark}
In~\cite{Haeupler_jacm} it is argued that in the vast majority of applications $\delta = O(n \log n ) $ and in many cases even linear in $n$.
\end{remark}

Following Theorem~\ref{main},  the proof of Theorem~\ref{clique_decomp} is  identical to the analogous result  of Hauepler, Saha and Srinivasan~\cite{Haeupler_jacm} for the Moser-Tardos algorithm and hence we omit it.

\paragraph{Fast Search for Flaws.}

Searching for occurring flaws efficiently can be a major obstacle in getting polynomial time algorithms, even in the case where convergence is guaranteed after a polynomial number of steps. Again, there is more than one approach one can follow 
to deal with this issue.

A first approach was introduced in~\cite{Haeupler_jacm}  where it is shown that Theorems~\ref{core_events} and~\ref{clique_decomp}, in the context of the variable setting,  can be combined into a single theorem that guarantees the existence of a Monte Carlo algorithm which runs in polynomial time, even in the presence of super-polynomially many flaws. The theorem assumes the existence of a polynomial size decomposition of the causality graph into cliques and, moreover, that the LLL conditions hold with an exponential slack. Using Theorem~\ref{main}, we can extend  this result in a straightforward way to our setting to get:

\begin{theorem}
Let $\mathcal{A} = (D,F, \sim,\rho)$ be a commutative algorithm such that the causality graph induced by $\sim$ can be partitioned into $n$, with potentially further edges between them. Assume there exist real numbers  $\{ x_i\}_{i=1}^{m}$ and $\epsilon$ in $(0,1)$ be such that for every $i \in [m]$ we have that
\begin{align*}
\gamma(f_i)^{1-\epsilon} \le x_i \prod_{ j \in \Gamma(i)} (1- x_j) \enspace.
\end{align*}
If we furthermore have that $ \log 1/ \delta \le \mathrm{poly}(n)$, where $\delta = \min_{ i \in [m] } x_i \prod_{j \in \Gamma(i) } (1 - x_j) $, then for every $\gamma \ge \frac{1}{ \mathrm{poly}(n)}$  the set $\{ i \in [m] \text{ s.t. } \gamma(f_i) \ge \gamma \}$ has size at most $\mathrm{poly}(n)$. There also exists a Monte Carlo algorithm that terminates after $O( \lambda_{\mathrm{init} } \frac{n}{\epsilon} \log \frac{n}{\epsilon^2} )$ steps and returns a perfect object with probability at least $1 - n^{-c}$, where $c$ is any desired constant.
\end{theorem}

In a  follow-up work~\cite{EnuHarris},  Harris and Srinivasan describe a general technique that yields efficient procedures for searching for flaws. The main building blocks of their technique is a ``witness tree lemma for internal states" and problem-specific, possibly randomized, data-structures that contain the flaws that are present in each state.  We refer the reader  to~\cite{EnuHarris} for more details, but we note that combining the proof of~\cite{EnuHarris} with the proof of Theorems~\ref{WitnessTreeLemma} and~\ref{main}, one can show that the  ``witness tree lemma for internal states" holds for commutative algorithms.

\section{Proofs Omitted from Section~\ref{Proofs}}\label{Misc}

\subsection{Proof of Lemma~\ref{witness_trees_sum}}

We show the following more general lemma. The claim follows by applying this general lemma with $\mathrm{List}(j) = \mathrm{Ind}( \Gamma(i))$ for every $j \in [m]$.
  \begin{lemma}\label{counting}
Assume that for every $i \in [m]$ were are given a set $\mathrm{List}(i) \subseteq 2^{[m]}$ and there exist positive numbers $\{ \psi_i \}_{i=1}^m$ such that for each $i$:
\begin{align*}
\frac{ \gamma(f_i) }{ \psi_i} \sum_{S \in \mathrm{List}(i)  }  \prod_{j \in S} \psi_j \le 1 \enspace.
\end{align*}
Let $\mathcal{L}_i$ be  the set of trees whose root is labelled by $i$ and such that  the set of labels of every node $v$ with labael $(v)= j$ is in $\mathrm{List}(j)$, for every $j \in [m]$. Then:
\begin{align*}
\sum_{ \tau \in \mathcal{L}_i  }  \prod_{ v \in V(\tau) } \gamma \left(  f_{ (v  )}  \right) \le \psi_i \enspace.
\end{align*}
\end{lemma}

 \begin{proof}[Proof of Lemma~\ref{counting}]
 
 To proceed, we use ideas  from~\cite{MT,PegdenIndepen}. Specifically, we introduce a branching process that produces only trees in   $\mathcal{L}_i$  and bound $\sum_{\tau \in \mathcal{L}_i  }   \prod_{v \in V(\tau)  }  \gamma \left(  f_{ (v  )} \right) $ by analyzing it.

 In particular, we start with a single node labelled by $i$. In each subsequent round each leaf $u$ ``gives birth" to a set of nodes whose set of (distinct) labels is a set  $S \in \mathrm{List}((u) )$ with probability proportional to $\prod_{j \in S} \psi_j $. It is not hard to see that this process creates every  tree in $ \mathcal{L}_i $  with positive probability. 
 To express the exact probability received by each $S \subseteq [m]$ we define
\begin{equation}\label{eq:d_def}
Q(S)  :=    \prod_{j \in S} \param_j
\end{equation}
and let $Z_{\ell} = \sum_{ S \in \mathrm{List}((u))  }  Q(S)$. Clearly,  each $S \in \mathrm{List}((v))$ receives probability equal to $\frac{ Q(S) }{ Z_{\ell}}$.  We now show the following lemma.

\begin{proposition}\label{branchingLemma_2}
The branching process described above produces every tree $\tau \in \mathcal{L}_i$ with probability 
\begin{align*}
p_{\tau} = \frac{ 1}{ \psi_i} \prod_{v \in V(\tau) } \frac{ \psi_{(v) } }{ \sum_{ S \in \mathrm{List}((v)) } \prod_{ j \in S }  \psi_j } \enspace.
\end{align*}
\end{proposition}

\begin{proof} 

For each tree $\tau \in \mathcal{L}_i$ and each node $v$ of $\tau$, let $N(v)$ denote the set of labels of its children. Then:
\begin{eqnarray*}
p_{\tau} & = &    \prod_{v \in V(\tau)} \frac{  Q(N(v))}{\sum_{S  \in \mathrm{List}((v)) }  Q(S)} \\
& = & \frac{1}{ \psi_i}  \prod_{v \in V(\tau)} \frac{\param_{(v)}  }{ \sum_{S \in \mathrm{List}((v))  }Q(S)}  \enspace .
\end{eqnarray*}
\end{proof}
Notice now that
\begin{eqnarray}
\sum_{\tau \in \mathcal{L}_i  } \prod_{v \in V(\tau)  }  \gamma( f_{[v]} ) & \le & \sum_{\tau \in \mathcal{L}_i  } \prod_{v \in V(\tau)  } \frac{  \psi_{(v) }}{ \sum_{ S \in \mathrm{List}((v)) } \prod_{ j \in S  }  \psi_j }  \label{prwto_step_counting} \\
						 & = & \psi_{i} \sum_{\tau \in \mathcal{L}_i  }  p_{\tau}  \label{deutero_step_counting}  \\
						 & =& \psi_{i} \nonumber \enspace,
\end{eqnarray}

where~\eqref{prwto_step_counting} follows by the hypothesis of Lemma~\ref{counting}  while~\eqref{deutero_step_counting} by Proposition~\ref{branchingLemma_2}.

\end{proof}

\subsection{Proof of Lemma~\ref{Kolmo27}}

Let $\mathbf{u}_1, \ldots, \mathbf{u}_m$ be the named  indices of flaws of $\mathbf{Q}(\Sigma)$ that occur in $U$ listed in the order of their appearance in $U$. We claim that $\mathbf{u}_1, \ldots, \mathbf{u}_m$ is a prefix of $U$. To see this, assume for the sake of  contradiction that $U =  \ldots \mathbf{w } \mathbf{ u}_i  \ldots  $ where $\mathbf{w} \notin  \mathbf{Q}(\Sigma)  $ and $\mathbf{u}_i \in \mathbf{Q}(\Sigma)$. Thus, $( \mathbf{w}, \mathbf{u}_i ) \notin \mathbf{E}(\Sigma) $ and 
so $w \nsim u_i $. Therefore, $(\mathbf{w},\mathbf{u}_i)$ is a swappable pair in $U$, which is a contradiction.

Note that the latter observation implies that $W(\Sigma) = (A,B,C) $ where $B = ( \mathbf{u}_1, \mathbf{u}_2, \ldots, \mathbf{u}_m ) $. It remains to show that $B$ is a subsequence of $Q_{\pi}(\Sigma) $. In particular, it suffices to show that for any
$i \in [m-1]$ the relative order of $\mathbf{u}_i$ and $\mathbf{u}_{i+1}$ in $B$ is the same as in $Q_{\pi}(\Sigma)$. Assume the opposite, i.e., $Q_{\pi}(\Sigma) = \ldots \mathbf{u}_{i+1} \ldots \mathbf{u}_i \ldots$. It has to be that $\mathbf{u}_i \sim \mathbf{u}_{i+1}$. For otherwise, $(\mathbf{u}_i, \mathbf{u}_{i+1} ) $ would be a swappable pair, contradicting the assumption. This means that $(\mathbf{u}_i, \mathbf{u}_{i+1} )  \in \mathbf{E}(\Sigma) $, implying  that $d(\mathbf{u}_i) > d( \mathbf{u}_{i+1} ) $. Recalling the definition of $Q_{\pi}(\Sigma) $ we see that $\mathbf{u}_i$ should be to the left of $\mathbf{u}_{i+1}$ in $Q_{\pi}(\Sigma)$, a contradiction.

\subsection{Proof of Lemma~\ref{Kolmo28}}

First we prove that the mapping is injective.  In particular, assume that two distinct trajectories $\Sigma_1, \Sigma_2 \in \mathcal{X}_p [W]  $ are transformed to the same trajectory $\Sigma \in \mathcal{X}_{p+1}[W] $. At least one of $\Sigma_1, \Sigma_2$ must have changed. Without loss of generality, assume $\Sigma \ne \Sigma_1$. The latter implies that  $W(\Sigma_1) = (A, \mathbf{w},\mathbf{y},B)$
with $ w\nsim y$ and $\Sigma_1$ is transformed to a trajectory  $\Sigma$ with $W(\Sigma) = (A, \mathbf{y}, \mathbf{w}, B )$. Notice that it cannot be the case that $\Sigma_2 = \Sigma$ since then $\Sigma_1$ and $\Sigma$ would both be in $\mathcal{X}_p[W]$ without following the same deterministic flaw choice strategy, a contradiction. Thus, $W(\Sigma_2) = (A, \mathbf{w}, \mathbf{y}, B) $, and $\Sigma$ was obtained from $\Sigma_2$ by swapping $\mathbf{w}$ and $\mathbf{y}$. Recalling that the  $\mathrm{Swap}$ operation is an injection, we get that $\Sigma_1 = \Sigma_2$.

We will assume that $\mathcal{X}_{p+1}[W]$ is not valid and reach a contradiction. The latter assumption implies that there should be trajectories $\Sigma, \Sigma' \in \mathcal{X}_{p+1}[W]$ such that 
\begin{eqnarray*}
 W( \Sigma)  &=&  (\mathbf{w}_1, \ldots , \mathbf{w}_{\ell} , \mathbf{w}, \mathbf{y} ,\ldots ) \\
 W( \Sigma')  &=& ( \mathbf{w}_1, \ldots , \mathbf{w}_{\ell}, \overline{\mathbf{w}} ,\overline{\mathbf{y}},  \ldots )  \enspace,
\end{eqnarray*}
 with $\mathbf{w} \ne \overline{\mathbf{w}}$ and the states in $\Sigma$ to the left of $\mathbf{w}$ match the corresponding states in $\Sigma'$ to the left of $\overline{\mathbf{w}}$. Here it is assumed that some of $\mathbf{w}, \mathbf{y}, \overline{\mathbf{w}}, \overline{\mathbf{y}}$ may equal $\emptyset$, which means they don't exist. It is also assumed that $\mathbf{w} = \emptyset $ also implies that $ \mathbf{y} = \emptyset$ and similarly for  $ \overline{ \mathbf{w} } $ and $\overline{\mathbf{ y} }$.

Let $\Sigma_1 $ and $\Sigma_2 $ be respectively the trajectories in $\mathcal{X}_p[W] $ that were transformed to $\Sigma $ and $\Sigma'$. Since $\mathcal{X}_p[W]$ is valid, at least one of them must have changed. Assume, without loss of generality, $\Sigma \ne \Sigma_1$. We know that (i) $\Sigma_1 $ and $\Sigma_2$ follow the same deterministic  flaw choice strategy, and they are not proper prefixes of each other, as wells as, that (ii) named flaws indices $\mathbf{w}_1, \ldots, \mathbf{w}_{\ell}, \mathbf{w}$ are distinct.

Some other useful facts to have in mind  (and which we will implicitly use) are that the first $\ell +1$ states of $\Sigma $ match those of $\Sigma'$, and also that swapping adjacent indices of flaws only affects the state between them in a deterministic way. Now there are four cases:

\begin{enumerate}[(a)]

\item The swapped pair in $\Sigma$ was $(\mathbf{w}_i, \mathbf{w}_{i+1})$ for $i \in [ \ell -1]$. Thus, 
\begin{align*}
W(\Sigma_1) = (  \mathbf{w}_1, \ldots, \mathbf{w}_{i+1}, \mathbf{w}_i, \ldots, \mathbf{w}_{\ell}, \mathbf{w},\mathbf{y}, \ldots ) \enspace.
\end{align*}
Using (i) and (ii), we conclude that $\Sigma' \ne \Sigma_2$ and, thus, 
\begin{align*}
W(\Sigma_2) = ( \mathbf{w}_1, \ldots, \mathbf{w}_{i+1}, \mathbf{w}_i, \ldots, \mathbf{w}_{\ell}, \overline{ \mathbf{w} }, \ldots, \mathbf{y}, \ldots ) \enspace.
\end{align*}
To see this, recall that swaps are applied at the same positions for trajectories in $\mathcal{X}_p$. Condition (i) now implies that $\mathbf{w} = \overline{ \mathbf{w} }$.

\item The swapped pair in $\Sigma$ was $( \mathbf{w}_{\ell}, \mathbf{w} )$, and so $W(\Sigma)  = (\mathbf{w}_1, \ldots,  \mathbf{y}, \mathbf{w}_{\ell}, \mathbf{w}, \ldots  )$. Thus, $ \Sigma' \ne \Sigma_2$ by (i) and (ii). Now condition (i) implies that $\mathbf{w} = \overline{ \mathbf{w} } $.

\item The swapped pair in $\Sigma$ was $( \mathbf{w}, \mathbf{y}  ) $, and so $ W(\Sigma_1) = ( \mathbf{w}_1, \ldots, \mathbf{w}_{\ell}, \mathbf{y}, \mathbf{w}, \ldots ) $ and $\mathbf{w} \in \mathbf{Q}(\Sigma_1)$. We now apply Lemma~\ref{Kolmo27} to $\Sigma_1$ and we notice that since there are no swappable pairs in $\Sigma_1$ to the right of $(\mathbf{y},\mathbf{w} ) $ it should be that $W(\Sigma_1) = (\mathbf{w}_1, \ldots, \mathbf{w}_{\ell},B,C)$ where $B$ is a subsequence of $Q_{\pi}(\Sigma_1) $ and $C$ does not contain named indices of flaws from $Q_{\pi}(\Sigma) $. Observe that $B$ should start with $\mathbf{w}$. Using (i) we get that $W(\Sigma_2) = (\mathbf{w}_1, \ldots, \mathbf{w}_{\ell}, \mathbf{y}, \ldots ) $ and again by Lemma~\ref{Kolmo27}: $W(\Sigma_2)  = ( \mathbf{w}_1, \ldots, \mathbf{w}_{\ell}, \mathbf{y}, \overline{B}, \overline{C} ) $ where $\overline{B}$ is a subsequence of $Q_{\pi}(\Sigma_2)$ and $\overline{C}$ does not contain named indices of flaws from $\mathbf{Q}(\Sigma_2)$.

By the definition of $\mathcal{X}_p[W]$ we know that $Q_{\pi}(\Sigma_1) = Q_{\pi}(\Sigma_2) = W$. The latter fact, along with the forms of $\Sigma_1, \Sigma_2$  imply that  $B$ should be a permutation of $\overline{B}$. Notice though that $B$ and $\overline{B}$ are subsequences of $W$ and, furthermore, all elements of $W$ are distinct. Therefore, it has to be that $B = \overline{B}$.

The latter observation implies that $ W(\Sigma_2) = (  \mathbf{w}_1, \ldots ,\mathbf{w}_{\ell}, \mathbf{y}, \mathbf{w}, \ldots ) $. Now since $(\mathbf{y}, \mathbf{w})$ is a swappable pair in $\Sigma_1$ it should also be a swappable
pair in $\Sigma_2$. Thus, $\Sigma' = ( \mathbf{w}_1, \ldots,  \mathbf{w}_{\ell}, \mathbf{w},\mathbf{y}, \ldots )  $. This means that $\mathbf{w} = \overline{ \mathbf{w} } $.

\item The swapped pair in $\Sigma$ was to the right of $( \mathbf{w}, \mathbf{y})$. In this case condition (i) implies that $\mathbf{w} = \overline{ \mathbf{w} }$.

\end{enumerate}

\section{Proofs Omitted from Section~\ref{Byproducts}}\label{Misc}

\subsection{Proof of Theorem~\ref{EntropyBound} }

To lighten the notation, let $u := \lambda_{\mathrm{init} } \sum_{S \in \mathrm{Ind}([m]) }  \prod_{j \in S}  \psi_j   $. For each $\sigma \in \Omega$, define a  flaw $f_{\sigma} = \{ \sigma \}$ and consider the extended algorithm that addresses it by sampling from $\mu$, as well as  the extended causality graph that connects  $f_{\sigma}$ with every flaw in $F$. Clearly, we have that $\gamma(f_{\sigma} ) = \mu(\sigma)$. Moreover, if the original algorithm is commutative, so is the extended one since the commutativity condition is trivially true for flaws  $ \{ f_{\sigma} \}_{ \sigma \in \Omega } $. Observe now that for  every  $\sigma \in \Omega$,  Theorem~\ref{main} yields $\nu(\sigma) \le \Pr[\sigma] \le u \cdot \mu(\sigma)$. Thus:
\begin{align*}
H_{\rho}[\nu] = \frac{1}{1- \rho} \ln \sum_{ \sigma \in \Omega } \nu(\sigma)^{\rho}  \ge \frac{1}{1- \rho} \ln \sum_{ \sigma \in \Omega } ( u \mu(\sigma) )^{\rho} =   \frac{1}{1- \rho}  \ln \sum_{ \sigma \in \Omega } \mu(\sigma)^{\rho} - \frac{\rho}{\rho-1} \ln u \enspace,
\end{align*}
concluding the proof.

 \subsection{Proof of Theorem~\ref{partial_avoiding}}
 
For each flaw $f_i$ we define a Bernoulli variable $Y_i$ with probability of success $p_i = \min \left\{ 1, \frac{ \psi_i }{ \zeta_i  \gamma(f_i) }  \right\}$. The sequence $\{Y_i \}_{i=1}^m$  and $\Omega$ induce  a new space $\Omega' = \Omega \times \{ 0,1 \}^m$ 
which can be thought as a ``labelled" version of $\Omega$, where each  state $\sigma$ is labelled with a binary vector of length $m$ whose $i$-th bit describes the state of $Y_i$. Similarly, measure $\mu$ and $ \{ Y_i  \}_{i=1}^{m}$ induce a measure $\mu'$
over $\Omega'$.  

In this new state space we introduce a new family of flaws $F' = \{ f_1', f_2', \ldots, f_m'  \}$, where $f_i'$ is defined as the subset of $\Omega'$ where $f_i$ is present and $Y_i = 1$. Consider now the algorithm $\mathcal{A}'$ that is induced by $\mathcal{A}$ as follows:  Each time we want to address flaw $f_i'$ we move in $\Omega$ by invoking $\mathcal{A}$ to address $f_i$ and also take a sample from   $Y_i$ to update the value of the $i$-th entry of the label-vector. 

It is not hard to verify that (i)   the charge of each flaw $f_i'$ is  $\gamma(f_i') = \gamma(f_i) p_i$ ;  (ii) any causality graph for $(\Omega, F, \mathcal{A})$ is also a causality graph for $(\Omega', F', \mathcal{A}')$   (and, in particular,  so is the one induced by $\sim$);    (iii) if $\mathcal{A}$ is commutative then so is $\mathcal{A}'$ ; and that (iv) the cluster expansion condition with respect to the causality graph induced by $\sim$ is satisfied. 

To conclude the proof, consider a flaw $f_i'$ and notice that in order for $f_i$ to be present in the output of $\mathcal{A}'$ it has to be the case that $Y_i = 0$. Notice now that Theorem~\ref{main} implies:
\begin{align*}
\nu(f_i \cap  Y_i =0)    \le    (1- p_i) \gamma(f_i)  \zeta_i     =  \max \left\{ 0, \gamma(f_i) \zeta_i - \psi_i \right\} \enspace.
\end{align*}

 \section{Proofs Omitted from Section~\ref{Applicatia}}\label{applicatia_appendix}

\subsection{Proof of Lemma~\ref{sufficient_proper_coloring}}\label{second_phase}

We will use the following LLL condition which is obtained from~\eqref{eq:LLL} by setting $\psi_i = \frac{ 2\mu(A_i) }{1 - 2 \mu (A_i) } $. It's proof can be found in~\cite{mike_book}.
\begin{proposition}\label{MikeLLL}
Let $(\Omega,\mu)$ be an arbitrary probability space and let $\mathcal{A} = \{A_1, \ldots, A_m \}$ be a set of (bad) events. For each $i \in [m] $ let $D(i) \subseteq ([m] \setminus \{i \}  )$  be such that 
$\mu( A_i \mid \cap_{j \in S} \overline{ A_j })  = \mu(A_i)$ for every $S \subseteq ( D(i) \cup \{i \} ) $.
If 
\begin{align*}
 \sum_{ j \in D(i) \cup \{i \} } \mu(f_j) < \frac{1}{4} \enspace \text{for each $i \in [m] $ }\enspace, 
 \end{align*}
 then the probability that none of the events in $\mathcal{A}$  occurs is strictly positive.
\end{proposition}

Let $\mu$ be the probability distribution induced by giving each $\mathrm{Blank}$ vertex $v$ a color from $L_v(\sigma) \setminus \mathrm{Blank} $ uniformly at random.  For any edge $e$ and color $c \in \bigcap_{ u \in e } L_u(\sigma) \setminus \mathrm{Blank} $ we define $A_{e,c}$ to be the
event that all vertices of $e$ receive $c$. We also define $\mathrm{Blank}(e)$ to be the set of vertices  of $e$ that are $\mathrm{Blank}$ in $\sigma$. Observe now that
\begin{align*}
\mu\left( A_{e,c} \right)  \le \frac{1}{ \prod_{ v \in  \mathrm{Blank}(e) } \left( |L_v(\sigma)| -1 \right)  }  \enspace.
\end{align*}
Furthermore, $A_{e,c}$ is mutually independent of all events with which it does not share a vertex. The lemma follows from Proposition~\ref{MikeLLL} (and can be made constructive using the Moser-Tardos algorithm)  as  flaws $B_v , Z_v$ are not present for every vertex $v \in V$ and so
 \begin{eqnarray*}
 \sum_{ v \in \mathrm{Blank}(e) } \sum_{ c \in L_v(\sigma) \setminus \mathrm{Blank} } \sum_{u \in T_{v,c}(\sigma)}  \mu \left( A_{ \{u,v \},c} \right) & =  &   \sum_{ v \in \mathrm{Blank}(e) } \sum_{ c \in L_v(\sigma) \setminus \mathrm{Blank} } \sum_{ u \in T_{v,c}(\sigma) }  \frac{1}{  (|L_v(\sigma)|-1) (|L_u(\sigma)| -1)   } \\
 	 & \le  & 2 \max_{ v\in \mathrm{Blank}(e) } \frac{1}{  \left( | L_v(\sigma)|-1 \right) \left(   L-1  \right) }    \sum_{ c \in L_v(\sigma) \setminus \mathrm{Blank} } | T_{v,c'}(\sigma) |  \\
	 & \le&   \frac{2}{10} \max_{ v \in \mathrm{Blank}(e) }   \frac{L \cdot  |L_v(\sigma) |}{(L-1) \cdot ( |L_v(\sigma)| -1)} \\
	 &   \le & \frac{1}{5} \left(  \frac{L  }{ L-1 } \right)^2   < \frac{1}{4} \enspace,
 \end{eqnarray*}
 for  large enough $ \Delta   $, concluding the proof.

\subsection{Proof of Lemma~\ref{combo}}

Recall the description of $\mathcal{A}_2$ from the proof of Lemma~\ref{sufficient_proper_coloring}. 

First, we show that  $\mathcal{A}_2$ is able to output at least $N q^{-(1-\alpha) n }  $ list-colorings with positive probability.  Let $\Omega_{\mathcal{A}_1}^*$ denote the set of flawless partial list-colorings algorithm $\mathcal{A}_1$ can output with positive probability, and note that, according to our assumption, $|\Omega^*_{\mathcal{A}_1}| =  N$. To see the idea behind the bound, observe that given two colorings $\sigma_1, \sigma_2 \in \Omega_{ \mathcal{A}_1}^*$,  applying $\mathcal{A}_1$ to each one of them is guaranteed to result in different full list-colorings unless  there is a way to start from $\sigma_1$ (respectively, from $\sigma_2$) and assign colors to  $\mathrm{Blank}$ vertices so that we reach $\sigma_2$ (respectively, to $\sigma_1$). In this  bad case we write  $\sigma_1 \bowtie \sigma_2$. Consider now the graph $H$ over $\Omega_{ \mathcal{A}_1 }^*$  in which two colorings $\sigma_1, \sigma_2$ are adjacent iff $\sigma_1 \bowtie \sigma_2$, and observe  that the size of any independent set of $H$ is a lower bound on the number of list-colorings $\mathcal{A}_2$ can output. Since we have assumed that every coloring in $\Omega_{ \mathcal{A}_1 }^*$   has at least $\alpha n$ vertices colored, we see that the maximum degree of $H$ is at most  $ D:= q^{(1-\alpha)n}-1$ and, therefore, there exists an independent set of size at least $\frac{|\Omega_{ \mathcal{A}_1 }^*|}{D+1} =  N q^{-(1-\alpha)n }$, concluding the proof of the first part of Lemma~\ref{sufficient_proper_coloring}.

Second, we show that $\mathcal{A}_2$ is able to output at least $\left(\frac{8L }{11 } \right)^{ (1-\alpha)n} $  list colorings with positive probability.  To do that, we will need the following theorem regarding the output distribution of the Moser-Tardos algorithm that was proved in~\cite{EnuHarris}, and which we rephrase here to fit our needs.

\begin{theorem}[\cite{EnuHarris}]\label{dense_counting}
Consider a constraint satisfaction problem on  a set of variables $\mathcal{V}$ and set of constraints $\mathcal{C}$. Assume we have a flaw $f_c$ for each constraint $c$, comprising the set of states that violate $c$. We are also given an undirected causality graph such that two constraints are connected with an edge iff they share variables. For each constraint $c$ define  
\begin{align*}
y_c = \left( 1 +  \psi_c \right)^{ \frac{1}{|\mathrm{var}(c)|}}  -1 \enspace,
\end{align*}
where $\mathrm{var}(c)$ denotes the set of variables that correspond to constraint $c$. Then:
\begin{align*}
\sum_{  S \in \mathrm{Ind}(\mathcal{C})  }\prod_{c\in S} \psi_c \le  \prod_{ v \in \mathcal{V} } \left( 1 + \sum_{ \substack{c \in \mathcal{C} \\  v \in \mathrm{var}(c) } }  y_c \right)  \enspace,
\end{align*}
where $\mathrm{Ind}(\mathcal{C} ) $ denotes the set of independent sets of $\mathcal{C}$.
\end{theorem}

Observe that the hypothesis implies that $\mathcal{A}_1$ can output   a flawless partial coloring $\sigma$ where exactly $\alpha n$ vertices are colored with positive probability. We apply $\mathcal{A}_2$ to $\sigma$, which recall that is an instantiation of the Moser-Tardos algorithm using the uniform measure $\mu$ over the  the cartesian product,  $\Omega'$, of the lists of non-$\mathrm{Blank}$ available colors of the $\mathrm{Blank}$ vertices of $\sigma$, and where we have a bad event  $A_{e,c}$ for any edge $e$ and color $c \in  \bigcap_{u \in e } L_u(\sigma) \setminus \{ \mathrm{Blank} \}$. Recall further that the general (and, thus, also the cluster expansion) LLL condition~\eqref{eq:lllcond} is satisfied with   $\psi_{e,c} = \frac{2\mu(A_{e,c}) }{ 1 - 2\mu(A_{e,c})}$.  Thus,  we can combine Theorem~\ref{EntropyBound} and Theorem~\ref{dense_counting} to get that $\mathcal{A}_2$ can output at least
\begin{eqnarray*}
|\Omega'| \left(  \prod_{  \substack{ v \in V  \\  \sigma(v) = \mathrm{Blank}} } \left( 1 + \sum_{ c \in L_v(\sigma) \setminus \mathrm{Blank}   } \sum_{ u \in T_{v,c}(\sigma)   }  y_{ \{ u, v\},c} \right) \right)^{-1}  \ge    \frac{L^{(1-\alpha)n }}{   \left( 1 +  \frac{3}{8} \right)^{(1-\alpha)n }  }  \ge  \left(\frac{8L }{11 } \right)^{ (1-\alpha)n} \enspace,
\end{eqnarray*}
 list-colorings with positive probability. To see this, notice that  for any $v \in V$ that is $\mathrm{Blank}$ in  $\sigma$, and sufficiently large $\Delta$,
\begin{eqnarray}
 \sum_{ c \in L_v(\sigma) \setminus \mathrm{Blank}   } \sum_{ u \in T_{v,c}(\sigma)   }  y_{ \{ u, v\},c}  & = &  \sum_{ c \in L_v(\sigma) \setminus \mathrm{Blank}  } \sum_{ u \in T_{v,c}(\sigma)   }  \left( \sqrt{ 1 +  \frac{2 \mu(A_{ \{u, v\}, c}) }{  1 - 2 \mu(A_{ \{u, v\}, c}) } } -1 \right) \nonumber \\
 														& \le &  \sum_{ c \in L_v(\sigma) \setminus \mathrm{Blank}   } \sum_{ u \in T_{v,c}(\sigma)   } 3 \mu(A_{ \{u,v \}, c} ) \nonumber \\
														& \le & 3  \cdot  \left( \frac{1}{2} \cdot  \frac{1}{4} \right)  =  \frac{3}{8} \label{teleiwsame} \enspace,
\end{eqnarray}
where to obtain~\eqref{teleiwsame} we perform identical calculations to the ones in Lemma~\ref{sufficient_proper_coloring}.

 \subsection{Proofs of Theorems~\ref{AECARES} and~\ref{LowWeight}}

\subsubsection{The Clique Lov\'{a}sz Local Lemma}

We first state the Clique Lov\'{a}sz Local Lemma  (reformulated to fit our setting) assuming as input a commutative algorithm $(F,C,\rho)$, where $C$ is a causality graph.
\begin{theorem}[The Clique Lov\'{a}sz Local Lemma]\label{Clikes}
Let $\{ K_1, K_2,\ldots, K_n \}$ be a set of cliques in  $C$ covering all the edges (not necessarily disjointly).
If there exists a set of vectors $\{ \mathbf{x}_1, \ldots, \mathbf{x}_n \}$ from $(0,1)^m$ such that the following condition
are satisfied: 
\begin{itemize}
\item for each $v \in [n]$: $\sum_{ i \in K_v } x_{i,v} < 1$;

\item for each $i \in [m]$, $\forall v$ such that $i \in K_v$:

\begin{align*}
\gamma(f_i) \le x_{i,v} \prod_{ u \ne v: K_u  \ni i } (1 - \sum_{j \in K_u \setminus \{ i \} } x_{j,u} )
\end{align*}

\end{itemize}
then:

\begin{enumerate}

\item $ \mu \left(  \bigcap_{i \in [m] } \overline{f_i}    \right) \ge  \prod_{v \in [m] } \left( 1 - \sum_{ i \in K_v} x_{i,v}  \right) > 0$

\item The algorithm terminates after an expected number of at most
\begin{align*}
\sum_{i \in [m] } \min_{ v : K_v \ni i } \frac{x_{i,v} }{ 1 - \sum_{j \in K_v \setminus \{ i \} } x_{j,v} }  \enspace,
\end{align*}
steps.
\end{enumerate}
\end{theorem}

We note that in~\cite{CliqueLLL} the authors first prove  the first part of their theorem, which implies the \emph{existence} of perfect objects, and then they invoke the results of~\cite{szege_meet} which imply that the Moser-Tardos algorithm converges under the Shearer's condition. In particular, they use the following fact,  which we will also find useful in our applications:
\begin{align}\label{clique_bound}
\frac{q_{ \{i \} }}{q_{\emptyset} } \le \min_{v  : K_v \ni i} \frac{x_{i,v} }{1 - \sum_{j \in K_v  \setminus \{ i \} } x_{j,v}  }   \enspace, \text{ for every $i \in [m] $ }.
\end{align}
To prove Theorem~\ref{Clikes} in our setting we can follow the same strategy, invoking Theorem~\ref{main} (in the Shearer's regime) instead of the main result of~\cite{szege_meet}. In fact, the proof of the first part of Theorem~\ref{Clikes} is identical to the one of~\cite{CliqueLLL} assuming that the input algorithm is a resampling oracle with respect to $\mu$ for each flaw $f_i$. In the general case, some extra work is required. Since in our application we will be using the Moser-Tardos algorithm (which is a resampling oracle for every flaw), and in the interest of brevity, we omit it.

Finally, we note that the authors provide a canonical way of decomposing the causality graph and applying the Clique Lov\'{a}sz Local Lemma in the variable setting of Moser and Tardos. Specifically, recall that in the variable setting the family of bad events is  determined by a set of independent discrete random variables $\{v_1, \ldots, v_m \}$,  and that two events are adjacent in the  dependency  graph whenever they share a variable. Thus, each variable $v$ forms a clique $K_v$ in the dependency graph
consisting of the events that are dependent on this variable.

\subsubsection{Finding Acyclic Edge Colorings}

We now recall the proof of~\cite{CliqueLLL} for Acyclic Edge Coloring. The proof is  the same as the one in~\cite{alon1991acyclic,mike_stoc} and the improvement comes from the use of Clique LLL instead of condition~\eqref{eq:LLL}.

\smallskip

Given $q$ colors and a graph $G$ with maximum degree $\Delta$ let $\Omega$ be the set of all edge $q$-colorings of $G$. We identify the two following types of flaws:

\begin{enumerate}

\item For a path $P $ of length $2$ let $f_{P}$ be the set of states in $\Omega$ in which $P$ is monochromatic.

\item For a cycle $C$ of even length let $f_C$ comprise the set of states in $\Omega$ in which $C$ is bicolored.

\end{enumerate}

Clearly a flawless element of $\Omega$ is an acyclic edge coloring of $G$. Our algorithm is the MT algorithm (the variables that correspond to each event are the edges of the path/cycle), $\mu$ is the uniform measure and $\theta = \mu$. Therefore:
\begin{align*}
\gamma(f_P)  &=  \mu(f_P) =  \frac{1}{q}  \enspace, \\ 
 \gamma(f_C)  & = \mu(f_C)  \le \frac{1}{q^{ |C|-2 }  } \enspace.
\end{align*}
Furthermore, two flaws are connected in the causality graph iff they share an edge. We now follow the canonical way of decomposing the causality graph into cliques by having one clique $K_e$ for each edge $e$ of $G$. Moreover:

\begin{itemize}

\item For each path $P \ni e$ of length $2$ we set $x_{f_P, e} = x_{P,e }  =  \frac{ c}{1 + \epsilon } \frac{1}{2\Delta -2} $

\item For a cycle $C$ of even length we set  $x_{f_C,e} = x_{C,e} = \frac{ c}{(1 + \epsilon)^{|C| /2} }  \frac{1}{ (\Delta-1)^{|C|-2} } $

\end{itemize}

for some positive $c, \epsilon$ to be determined later. 

Observe that the number of cycles of length $2 \ell$, where $ \ell \ge 2$, that contain any given edge $e$ is at most $(\Delta-1)^{2 \ell -2 }$, while the number of paths  of length $2$ that contain $e$ is at most $2 \Delta - 2$. Thus, it suffices to show that for 
each edge $e$ and each path of length $P $ and cycle $C$ of length $2 \ell$ that contain $e$ we have that:
\begin{eqnarray*}
\gamma(f_P)	& \le & x_{P,e} \prod_{e' \in P \setminus \{ e \} } \left(1 -  \sum_{j \in K_e} x_{j,e'} \right)    \enspace, \\
\gamma(f_C)	& \le  &x_{C,e} \prod_{e' \in C \setminus \{ e \} }  \left( 1 - \sum_{j \in K_e} x_{j,e'} \right) \enspace.
\end{eqnarray*}
which is
\begin{eqnarray*}
\frac{1}{q}	& \le &  \frac{ c}{1 + \epsilon } \frac{1}{2\Delta -2}    \left(  1 -  c \sum_{j =1 }^{\infty}(1+ \epsilon)^{-j} \right)     \enspace, \\
\frac{1}{q^{2\ell-2}}	& \le  &\frac{ c}{(1 + \epsilon)^{\ell} }  \frac{1}{ (\Delta-1)^{2\ell-2} } \left(  1 -  c \sum_{j =1 }^{\infty}(1+ \epsilon)^{-j} \right)^{2 \ell -1}   \enspace.
\end{eqnarray*}
The latter imply that  for the Moser-Tardos algorithm to converge it suffices to have  $ c < \epsilon$ and also
\begin{eqnarray*}
\frac{q}{\Delta-1} &\ge&  \max \left\{  \frac{2}{c} (1+ \epsilon) \frac{\epsilon}{\epsilon-c} , \max_{ \ell \ge 2}  \left\{ (1+\epsilon)^{\frac{\ell}{2\ell-2}  } \frac{1}{ c^{ \frac{1}{2\ell-2} } }  \left( \frac{ \epsilon }{ \epsilon - c } \right)^{ \frac{2\ell-1}{ 2 \ell -2}}  \right\}  \right\}  \\
			 & \ge &\max \left\{  \frac{2}{c} (1+ \epsilon) \frac{\epsilon}{\epsilon-c} ,    \frac{(1+\epsilon)}{ \sqrt{c} }  \left( \frac{ \epsilon }{ \epsilon - c } \right)^{ \frac{3}{ 2}}  \right\}   \ge 8.59 \enspace.
\end{eqnarray*}

for $ \epsilon = 2.05869$ and $c = 0.8282 $.

\smallskip
We now show that the expected running time of our algorithm is polynomial. To do so, we will need to be careful in our flaw choice strategy and also deal with the fact that the number of flaws is super-polynomial.

As far as the flaw choice strategy is concerned, we choose to give priority to flaws of the form $f_{P}$, whose presence in the current state can be verified in polynomial time in the number of edges of $G$. This means we only address bichromatic cycles when the underlying edge-coloring is proper. Note that in order to find bichromatic cycles in a properly edge-colored graph we can simply consider each of the ${ q \choose 2}$ pair of distinct colors and seek cycles in the subgraph of the correspondigly colored edges.

One way to  address the fact that the number of flaws is super-polynomial is to invoke Theorem 7 of~\cite{szege_meet} that states that whenever the LLL condition is satisfied with a multplicative slack (notice that we have showed that the Clique LLL is satisfied  with $q \ge 8.59)$ then the expected number of resampling of the Moser-Tardos algorithm is polynomial in the number of the independent random variables (in our case, the edges of $G$) with constant probability.  
The probability of polynomial convergence can be boosted  by repetition. Overall, we have shown the following theorem.

\begin{theorem}
Given a graph $G$ with maximum degree $\Delta$ and $ q \ge 8.6 $ colors there exists an algorithm that outputs an acyclic edge coloring in expected polynomial time.
\end{theorem}

\subsubsection{Proof of Theorem~\ref{AECARES}}

To prove Theorem~\ref{AECARES} we need to estimate $\sum_{  S \in \mathrm{Ind}(F)  }\prod_{f \in S}  \frac{q_{ \{ f\} } }{q_\emptyset }$ per Theorem~\ref{EntropyBound} and Remark~\ref{shearer_remark} (we slightly abuse the notation and indicate $\psi$'s using flaws instead of indices of flaws). To do so, we will use Theorem~\ref{dense_counting} where we replace $\psi_c$ with $q_{ \{ c \}  }/ q_{\emptyset}$.


Using the result of the previous subsection along with~\eqref{clique_bound} we have that for each path $P$ of length two and each cycle $C_{\ell}$ of length $\ell$:
\begin{eqnarray*}
y_{P} & \le & \left( 1 + \frac{c}{1+ \epsilon}  \frac{1}{2 ( \Delta-1) } \frac{ \epsilon}{\epsilon-c}  \right)^{ \frac{1}{ 2}} - 1  <  \left( 1 + \frac{1}{ \Delta-1}   \right)^{1/2}   -1 < \frac{ 1}{ 2(\Delta -1)}  \\
y_{C_{\ell}} & \le &  \left( 1 +  \frac{c }{ (1 + \epsilon)^{\ell} } \frac{1}{ \left( \Delta -1 \right)^{2 \ell -2 }}   \right)^{\frac{1}{2 \ell} } -1 <  \left( 1 +  \frac{3^{-\ell}}{ \left( \Delta -1 \right)^{2 \ell -2 }}   \right)^{\frac{1}{2 \ell} }  -1 < \frac{1}{ (2(\Delta-1))^{2\ell -2}}
\end{eqnarray*}
and, thus, we get:
\begin{align*}
\sum_{  S \in \mathrm{Ind}(F)  }\prod_{f \in S}  \frac{q_{ \{ f\} } }{q_\emptyset }  \le  \prod_{e \in E} \left( 1 + 2 (\Delta-1) \frac{1}{2(\Delta-1) } + \sum_{i =2}^{\infty} ( \Delta-1)^{2\ell-2}  \frac{1}{ (2(\Delta-1))^{2\ell -2}}  \right)   <  4^{|E|} \enspace. 
\end{align*} 
Now Theorem~\ref{EntropyBound} and the fact that $|\Omega | = q^{|E|}$ conclude the proof.

\subsubsection{Proof of Theorem~\ref{LowWeight}}

For a vertex $v$ and a function $W_v$ let $N(v)$ denote the set of edges adjacent to $v$ and $\mathcal{A}_v$ denote the set of possible edge $q$-colorings of the edges in $N(v)$. For $\alpha \in \mathcal{A}_v$ let $E_v(\alpha)$ be the subset of $\Omega$
whose elements assign $\alpha$ to the edges in $N(v)$.  Moreover, consider the resampling probability distributions induced by the Moser-Tardos algorithm for $E_v(\alpha)$.   Observe that using~\eqref{clique_bound} we get (again, slightly abusing the notation):
\begin{eqnarray*}
\sum_{ S \in \mathrm{Ind}( \Gamma(E_v(\alpha) ) ) } \prod_{f \in S} \frac{ q_{ \{ f\} }}{q_{\emptyset} } & \le&  \prod_{e \in N(v) }  	\left( 	1 + \sum_{ P \ni e } x_{P,e} + \sum_{ C \ni e }  x_{C,e}   \right)	\\
																		& \le & \prod_{e \in N(v) } \left( 1  +   \left( \frac{c}{1+ \epsilon}   + \sum_{ \ell =2}^{\infty}  \frac{c}{(1+\epsilon)^{\ell}}   \right) \frac{\epsilon }{ \epsilon-c }   \right) \\
																		& < & 1.3^{\Delta}  \enspace.
\end{eqnarray*}
Now applying Theorem~\ref{main} we get:
\begin{eqnarray*}
 \ex[ W_v ]   &= &  \sum_{ \alpha \in \mathcal{A}_v } \Pr \left[  E_v( \alpha) \right]  W_v(\alpha)   \\
 		 & < &  1.3^{\Delta}  \sum_{ \alpha \in \mathcal{A}_v }  \mu(E_v(\alpha) ) W_v(\alpha) \\
		 & = & 1.3^{\Delta }  \cdot \ex_{ \phi \sim \mu } [W_v(\phi) ] \enspace.
 \end{eqnarray*}
  Linearity of expectation concludes the proof.

\end{document}